\documentclass[11pt,letterpaper]{article}
\pdfoutput=1

\RequirePackage[OT1]{fontenc}
\RequirePackage[colorlinks,citecolor=blue,urlcolor=blue]{hyperref}
\usepackage{amssymb,amsmath,amsthm,color,outlines,subfigure,comment}
\usepackage[small]{caption}
\usepackage{graphics}
\usepackage{graphicx}
\usepackage{stackrel}
\usepackage[round]{natbib}

\newcommand{\bft}{{\bf{t}}}

\newcommand{\bfx}{{\bf{x}}}

\newcommand{\bfS}{{\mathbf{S}}}
\newcommand{\bfT}{{\mathbf{T}}}

\newcommand{\bfX}{{\mathbf{X}}}

\newcommand{\BE}{{\mathbb{E}}}

\newcommand{\BF}{{\mathbb{F}}}

\newcommand{\BH}{{\mathbb{H}}}
\newcommand{\BI}{{\mathbb{I}}}
\newcommand{\BJ}{{\mathbb{J}}}
\newcommand{\BK}{{\mathbb{K}}}

\newcommand{\BP}{{\mathbb{P}}}
\newcommand{\BR}{{\mathbb{R}}}
\newcommand{\BS}{{\mathbb{S}}}

\newcommand{\BV}{{\mathbb{V}}}

\newcommand{\bfgamma}{{\boldsymbol{\gamma}}}
\newcommand{\bfmu}{{\boldsymbol{\mu}}}

\newcommand{\cD}{\mathcal{D}}

\newcommand{\cF}{\mathcal{F}}

\newcommand{\cI}{\mathcal{I}}
\newcommand{\cM}{\mathcal{M}}

\newcommand{\cS}{\mathcal{S}}
\newcommand{\cX}{\mathcal{X}}

\def\baq#1\eaq{\begin{align}#1\end{align}}
\def\ban#1\ean{\begin{align*}#1\end{align*}}
\newcommand{\draro}{ \stackrel{{\mathcal D}}{\Rightarrow} }

\def\bredbf#1\eredbf{{\color{red}{\bf ???? #1 ????}}}
\def\BeginRedComment#1\EndRedComment{({\color{red}{\bf Comment:} {\it #1}}) }
\def\BeginBlueComment#1\EndBlueComment{({\color{blue}{\bf Comment:} {\it #1}}) }

\DeclareMathOperator*{\ve}{vec}
\DeclareMathOperator*{\diag}{diag }

\makeatletter
\newcommand{\opnorm}{\@ifstar\@opnorms\@opnorm}
\newcommand{\@opnorms}[1]{%
  \left|\mkern-1.5mu\left|\mkern-1.5mu\left|
   #1
  \right|\mkern-1.5mu\right|\mkern-1.5mu\right|
}
\newcommand{\@opnorm}[2][]{%
  \mathopen{#1|\mkern-1.5mu#1|\mkern-1.5mu#1|}
  #2
  \mathclose{#1|\mkern-1.5mu#1|\mkern-1.5mu#1|}
}
\makeatother

\newtheorem{Theorem}{Theorem}[section]
\newtheorem{Lemma}[Theorem]{Lemma}
\newtheorem{Corollary}[Theorem]{Corollary}
\newtheorem{Proposition}[Theorem]{Proposition}

\theoremstyle{definition} \newtheorem{Definition}[Theorem]{Definition}

\begin{document}

\title{On Weighted Multivariate Sign Functions}
\date{}
\author{
Subhabrata Majumdar\thanks{Currently at AT\&T Labs Research. Email: {\tt subho@research.att.com}}\\
and\\
Snigdhansu Chatterjee\thanks{Corresponding author. Email: {\tt chatt019@umn.edu}}\\
	University of Minnesota, Minneapolis, MN, USA
}
\maketitle

\noindent\textbf{Abstract}: 
Multivariate sign functions are often used for robust estimation and inference. We propose using data dependent weights in association with such functions. The proposed weighted sign functions  retain desirable robustness properties, while significantly improving efficiency in estimation and inference compared to unweighted multivariate sign-based methods. Using weighted signs, we demonstrate methods of robust location estimation and robust principal component analysis. We extend the scope of using robust multivariate methods to include robust sufficient dimension reduction and functional outlier detection. Several numerical studies and real data applications demonstrate the efficacy of the proposed methodology.

\vspace{.5cm}


\section{Introduction}
\label{Sec:Introduction}
 
 Given a point $\mu$ in a normed linear space $\cX$ with norm denoted by $| \cdot |$, 
 the \textit{generalized sign function} 
 $S : \cX \times \cX \mapsto \cX$ with center $\mu$ is  
 defined as
\baq
S (x; \mu) =  \left\{ 
\begin{array}{ll}
| x - \mu|^{-1} ( x - \mu ) &  \text{ if } x \ne \mu, \\
0 & \text{ if } x = \mu.
\end{array}
\right.
\label{eq:GSign}
\eaq
This is a functional and multivariate  generalization of the real-valued 
\textit{sign function}, that 
takes the values one, negative one or zero if the point $x \in \BR$ is to the right, left 
or equal $\mu \in \BR$ respectively. This generalized sign function was  
introduced by \cite{ref:JNonpara95201_MottonenOja95} for $\cX = \BR^{p}$, 
the $p$-dimensional  real Euclidean space.

The function $S$ maps $\mu$ to the origin and all other points of $\cX$ 
to the unit sphere  ${\cS}_{0; 1} = \{ x \in \cX: | x | = 1 \}$. 
Given a dataset  $\{ X_{i} \in \BR^{p}: i =1, \ldots, n \}$,
that we collect together in 
the $n \times p$ matrix $\bfX = ( X_{1}: \ldots : X_{n} )^{T}$,
an approach for robust estimation and inference in multivariate data starts by evaluating  
\eqref{eq:GSign} on each observation, thus defining 
$S_{i} = S ( X_{i}; \mu_i)$ with respect to some suitable center 
$ \mu_i \in \BR^{p}$, and then using these for robust location and scale estimation and 
inference, including inference for $ \mu_i$ \citep{ref:Test991_Locantoreetal, ref:OjaBook10, ref:JASA151658_WangPengLi}.
Suppose $\bfS = ( S_{1}: \ldots : S_{n} )^{T} \in \BR^{n \times p}$. 
If the data $\{ X_{i} \in \BR^{p}: i =1, \ldots, n \}$ are independent, identically 
distributed (hereafter, i.i.d.) from an elliptically symmetric distribution, then the 
eigenvectors of $\BE ( X_{1} - \tilde{\mu}) ( X_{1} - \tilde{\mu})^{T}$
and of $\BE S_{1} S_{1}^{T}$ are the same for suitable centering 
parameter $\tilde{\mu} \equiv \mu_i$, that is, the population principal components 
from the original data and from its sign transformations are the same 
\citep{ref:SPL12765_Taskinenetal}. However, valuable information is lost in the form of magnitudes of sample points. As a result, spatial sign-based procedures suffer from low efficiency. For example, eigenvector estimates obtained from the covariance matrix of $S$ are asymptotically inadmissible \citep{ref:Biometrika14673_MagyarTyler} and  Tyler's M-estimate of scatter \citep{ref:AoS87234_Tyler} has uniformly lower asymptotic risk. 

In this paper, we propose to alleviate this low efficiency problem, by associating 
a data-driven weight $W_{i}$ with the generalized sign $S_{i}$, that can be used 
to adaptively trade-off between efficiency and robustness considerations in any given 
application. We demonstrate the utility of using the proposed \textit{weighted generalized sign} functions in a number of problems of current interest, including robust estimation of location and scatter.

Specifically, we propose using product of the generalized sign function and a weight function derived as a transformation of a data-depth function \citep{ref:DIMACS061_Serfling, ref:AoS00461_ZuoSerfling}. Like data-depth functions, the weight functions used in this paper are non-negative reals defined over $\cX \times \cF$, where $\cF$ is a fixed family of probability measures. For every choice of parameters $\mu \in \cX$ and $\BF \in \cF$, in this paper
\begin{align}
R (X_{i}; \mu, \BF) = S( X_{i}, \mu) W( X_{i}, \BF).
\label{eqn:rank}
\end{align}
is used as a robust surrogate for observation $X_{i}$.
Notice that for the trivial choice $W( x, \BF) = | x - \mu_{\BF} |$, 
 $\mu = \mu_{\BF}$,  we get $R (X_{i}; \mu, \BF) = X_{i}-\mu_{\BF}$, the original centered observations. 
 With the other trivial choice of  $W( x, \BF) \equiv 1$, we get the 
 generalized  sign $R (X_{i}; \mu, \BF) = S(X_{i}, \mu) = S_{i}$. 
 However, in this paper we illustrate how using other weight functions
 can lead to interesting robustness and efficiency trade-offs in a variety of situations. 
 The various technical conditions and assumptions that we impose on the 
 weight function $W( x, \BF)$ are valid for weights derived 
 from three well-known data depth functions: the \textit{half-space depth}, 
 the \textit{Mahalanobis depth}, and the \textit{projection depth}. Note that for i.i.d. data, there are three parameters involved here: the location parameter $\mu$ in 
 \eqref{eq:GSign}, the distribution $\BF$ used for the weight function, and the distribution $\BF_{X}$ of $X_{1}$. For clarity and to mirror the contexts of how data depth has been used in the literature \citep{LiuPareliusSingh99, ref:DIMACS061_Serfling}, we fix $\BF = \BF_{X}$ for this paper, although in Section~\ref{Sec:WSQuantiles} we briefly remark on the case when these  distributions are different. Also note that 
 in many problems of interest, $\BF$ is 
unknown and  data-depths are computed using $\BF_{n}$, however, under very standard regularity conditions (for example, see assumptions \textbf{B1-B3} below), the 
properties of $W (x, \BF)$ and $W (x, \BF_{n})$ are close enough for both asymptotic theory and practical applicability.  
 
\begin{figure}[t!]
\begin{center}
\includegraphics[width=\textwidth]{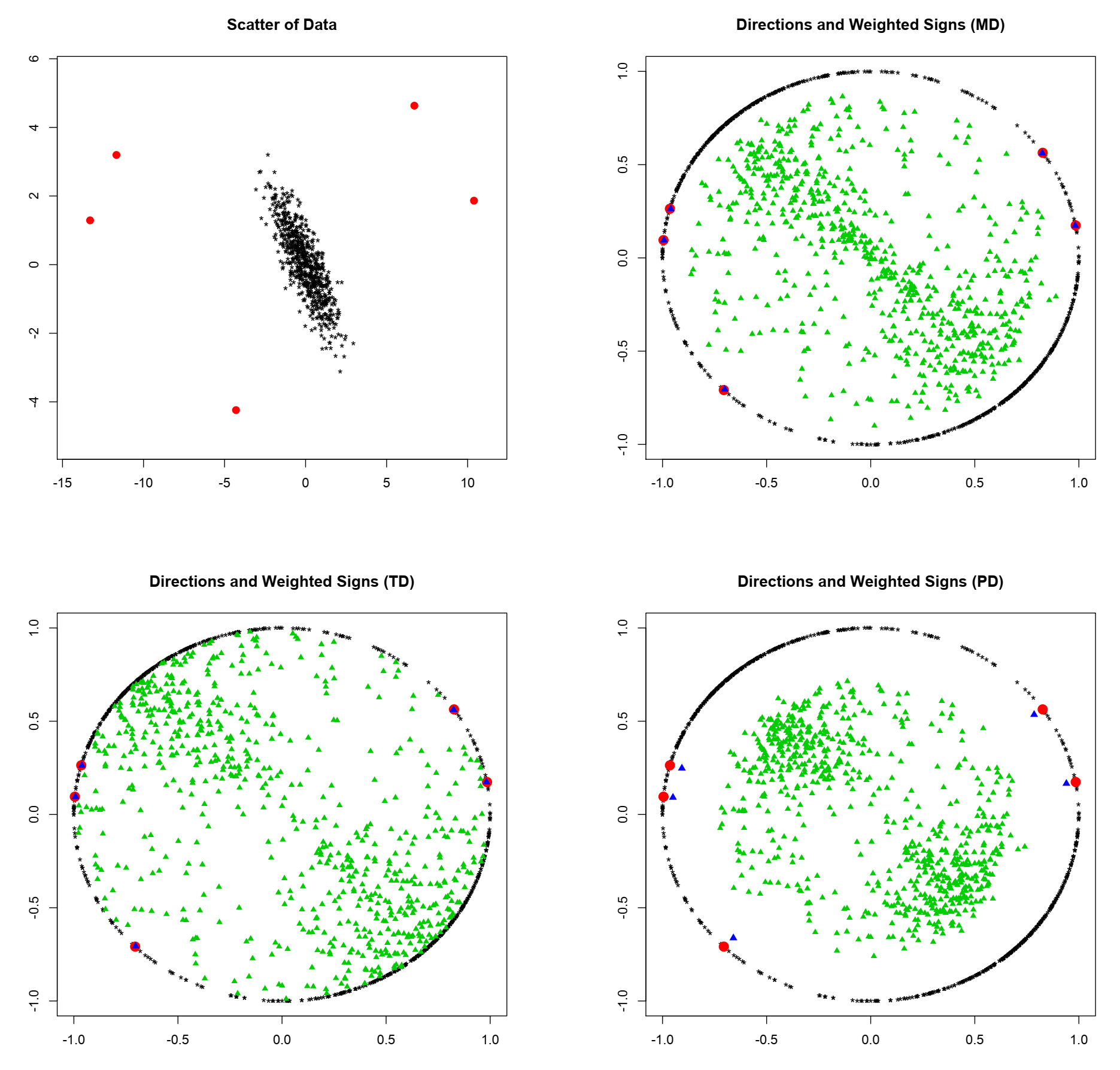}


\caption{
An illustrative bivariate scatter plot in the top left panel where the outliers are 
identified with red circles, 
and generalized signs from the same data (black points on the 
unit radius circle, outliers are red points) in the other panels. 
In the top right (bottom left, bottom right) 
panel, weighted signs from the same data with weights obtained using Mahalanobis depth 
(Tukey depth, projection depth respectively) are presented as green triangles 
(outliers are identified by blue triangles). 
}
\label{fig:Fig1}
\end{center}
\end{figure}

We primarily focus on the task of \textit{robust dispersion/scatter estimation} and 
\textit{robust principal component analysis} in this paper. Figure~\ref{fig:Fig1} presents an illustrative example of bivariate data with outliers in the top left panel, where the outliers are marked with red points. In the 
other panels, the generalized sign values of the same data are presented as black points 
on the unit circle, with the outliers again marked with red points. Notice that the 
black points from either the top right or bottom panels have very similar eigenvector 
structure as the original data without the outliers. The green and blue triangles 
are examples of 
the proposed \textit{weighted sign} values: the top right (respectively, bottom left, 
bottom right) panels depict these values where the weights have been generated using 
Mahalanobis depth (respectively Tukey's half-space depth and the projection depth). The 
blue triangles are the weighted sign values of the outliers. Notice that the eigenvectors 
from the weighted signs also capture the pattern from the original data without the 
outliers.

There are two unknown quantities in the generalized sign function as defined in \eqref{eq:GSign}: $\mu$ and $\BF$, To estimate dispersion and its eigen-structure robustly, we must start with a robust estimator for $\mu$. In Section~\ref{Sec:WSQuantiles} we present the case for \textit{weighted spatial quantiles}, which can be defined and studied in very general spaces 
$\cX$. One special case of this is the {\it weighted spatial median}. As a location estimator, it has several interesting robustness properties and can be shown to be more efficient that some existing robust location estimators, thus making it a perfect candidate to estimate $\mu$. 

Following that, we restrict to the $\BR^{p}$ for fixed $p$, and present detailed discussions on our primary proposal for a  robust measure of dispersion in Section~\ref{Sec:WSDispersion1}, followed by a proposed \textit{ affine equivariant version} of it in Section~\ref{Sec:WSDispersion2}, robust estimation of eigenvalues and a third robust estimator for dispersion in Section~\ref{Sec:Eigen}, and a thorough study of robustness and efficiency using influence functions in Section~\ref{Sec:RE_Dispersion}. We then report multiple simulation-based numeric studies in Section~\ref{Sec:Simulation}, present several real data examples in Section~\ref{Sec:DA}, and concluding remarks in Section~\ref{Sec:Conclusion}.

In the rest of this paper, all finite-dimensional vectors are column vectors, and for a vector or matrix $a$, the notation $a^{T}$ stands for its transpose. The Gaussian distribution with mean $\mu$ and variance $\Sigma$ is denoted by $N (\mu, \Sigma)$. The identity matrix  is denoted by $\BI$, with or without a subscript to denote its dimension. The notations $A^{-1}, det (A), \lambda_{\min} (A), \lambda_{\max} (A)$ respectively stand for the inverse, determinant, minimum and maximum eigenvalues of a matrix $A$, whenever these are well-defined. For a scalar or vector valued random variable $Y$, $\BE Y$ denotes its expected value, while $\BV Y$ denotes its variance or covariance matrix.
\section{The Weighted Spatial Median}
\label{Sec:WSQuantiles}

Suppose the open unit sphere in $\cX$ is given by 
$\textrm{int} {\cX}_{0; 1} = \{ x \in \cX: | x | < 1 \}$, and let
 $u \in \textrm{int} {\cX}_{0; 1}$. We also fix the set of probability measures $\cM$, 
 and select $\BF \in \cM$. Consider a random element $X \in \cX$, and define the 
 function $\Phi (q; X, u, \BF) = W (X, \BF) \bigl\{ | X - q | 
+ \langle u, X - q \rangle \bigr\}$. 
We define the $(u, \BF)$-th \textit{weighted spatial quantile} of $\cX$ 
as the minimizer $q (u, \BF) \in \cX$ of the expectation of $\Phi (q; X, u, \BF)$, 
that is
\baq
\Psi (q; u, \BF) = \BE \Bigl[ W (X, \BF) \bigl\{ | X - q | 
+ \langle u, X - q \rangle \bigr\} \Bigr] = \BE \Phi (q; X, u, \BF).
\label{eq:WeightedQuantile}
\eaq
%
\noindent
This is a natural generalization of the spatial median \citep{ref:JASA96862_Chaudhuri, ref:Biometrika48414_Haldane,ref:AoS97435_Koltchinskii} ($W(X, \BF) \equiv 1$ and $u = {\bf 0}_p$), or more general spatial quantiles \citep{ref:AoS17591_Cardotetal_Median_HilbertSpace, ref:AoS141203_ChakrabortyChaudhuri_Banach_Quantile, ref:Bernoulli152308_Minsker_Median_Banach} ($W(X, \BF) \equiv 1$). We assume that 
$\Phi (q; X, u, \BF)$ is convex in $q$ for $\BF$-almost all values of $x \in \cX$.

In what follows, for brevity we elaborate only the case of the \textit{weighted spatial median} (thus $\Psi (q; 0, \BF) = \BE \bigl[ W (X, \BF) | X - q | \bigr] $) for the case of finite dimensional $\cX$. Specifically, we demonstrate the utility of using the {\it weighted} spatial median as opposed to using the traditional, unweighted versions found in literature. The analysis and computation of the general weighted spatial quantiles will largely follow by extending the results of the above cited references, and we postpone details to a future study.

Assume that we have independent and identically distributed observations $X_{1}, \ldots, X_{n} \in \cX$, and the sample weighted spatial median is computed by minimizing $\Psi_{n} (q; 0, \BF) = \sum_{i = 1}^{n} W (X_{i}, \BF) | X_{i} - q | $, and is denoted by $\hat{q}_{nW}$, the second subscript is in acknowledgement that the weight function is used. We denote the unweighted version of this estimator, ie, the case where $W (X, \BF) \equiv 1$ as $\hat{q}_{n}$. Assume the following technical conditions:

\begin{description}
\item[(A1)] $\Psi (q; 0, \BF)$ is finite for all $q \in \cX \subseteq \BR^{p}$ 
and has a unique minimizer $q_{0}$. 

\item[(A2)] $\Psi (q; 0, \BF)$ is twice differentiable at $q_{0}$ and the second derivative is positive definite. 
 
\item[(A3)] $ {\frac{\partial^{2}} {\partial q^{2}}} \Psi (q; 0, \BF)$
exists for all $q$ in a neighborhood of $q_{0}$, and we use the notations
\ban
\Psi_{1W}   & = 
\Bigl( {\frac{\partial} {\partial q}} \Psi (q_{0}; 0, \BF) \Bigr)
 \Bigl( {\frac{\partial} {\partial q}} \Psi (q_{0}; 0, \BF) \Bigr)^{T}, \\
\Psi_{2W}   & = {\frac{\partial^{2}} {\partial q^{2}}} \Psi (q_{0}; 0, \BF). 
\ean
 \end{description}
 These assumptions are very broad-based and general. The first one essentially requires 
 the existence of a population parameter, the second one requires that the minimization 
 approach is meaningful in the population, and the third one essentially requires that the 
 weight function has a finite variance. No further restrictions are placed on the tuning parameter $\BF$ or the choice of the weight function. 
 
 \begin{Theorem}\label{Thm:WSMedian}
 Under assumptions [A1]-[A3], we have 
 \ban 
 n^{1/2} (\hat{q}_{nW} - q_{0}) & \draro N(0, \Psi_{2W}^{-1} \Psi_{1W} \Psi_{2W}^{-1}).
 \ean
 \end{Theorem}

Thus, under very standard regularity conditions, the sample weighted spatial median 
is consistent and is asymptotically normal. Theorem~\ref{Thm:WSMedian}  can be proved in several different ways. Here we use techniques following \cite{ref:AoS891631_Haberman, ref:AoS921514_Niemiro}. Specifically, following Theorem 4 in \cite{ref:AoS921514_Niemiro}, which traces back to \cite{ref:AoS891631_Haberman} with slightly relaxed conditions, we get
\begin{align*}
n^{1/2} (\hat{q}_{nW} - q_{0}) =
- \frac{\Psi_{2W}}{\sqrt n} 
\sum_{i=1}^n \nabla \Phi (q; X_{i}, u, \BF)
+ o_P(1),
\end{align*}
where $\nabla \Phi (q; X_{i}, u, \BF)$ is any measurable subgradient of 
$\Phi (q; X_{i}, u, \BF)$. 
Theorem~\ref{Thm:WSMedian} follows by applying the central limit theorem.

\paragraph{Remark.} Note that the technical conditions for the result presented in Theorem~\ref{Thm:WSMedian} is one of several alternatives that can conceived, and the scope of this result is broader than what is presented above. First, note that if $\BF$ 
and $\BF_{X}$ are different and the weights are not a function of $q$, 
a situation that may arise in hypothesis testing problems where the weights are based 
on the null distribution, the convexity of $\Phi (q; X_{i}, u, \BF)$ follows automatically 
and is not an assumption. Second, even if $\Phi (q; X_{i}, u, \BF)$ is not convex but sufficiently smooth, we can have a central limit theorem, for example, by using techniques similar to \cite{ref:CBose_AoS05414}. 
Choices of $\BF$ other than $\BF_{X}$, e.g. \cite{StatPaper18},  may lead to interesting interpretations of $W(. ,\BF)$ and the resulting location estimator and will be explored further in future.

\subsection{Asymptotic efficiency of weighted spatial median}
Let $V_{W} = \Psi_{2W}^{-1} \Psi_{1W} \Psi_{2W}^{-1}$ be the asymptotic 
variance of $\hat{q}_{n,W}$ from Theorem~\ref{Thm:WSMedian}, where we use the 
subscript ``${}_{W}$'' to denote that this depends on the weight function. 
We use the notation 
$V_{1}$ for the case where $W (x, \BF) \equiv 1$, that is, all weights are one.
The asymptotic relative efficiency of two statistics is the $p$-th root of 
the reciprocals of their determinants. That is, 
\ban 
ARE (\hat{q}_{nW}, \hat{q}_{n}) = \Bigl\{\frac{det(V_{1})}{det(V_{W})} \Bigr\}^{1/p}.
\ean
One easy result from Theorem~\ref{Thm:WSMedian} is that under reasonable conditions the
asymptotic relative efficiency of the weighted spatial median over the unweighted version
is always greater than one. We document this in the following corollary:

\begin{Corollary}
\label{Cor:ARE}
Assume that the weight function $W (X, \BF)$ is bounded above by some $W_{\max} > 0$, and the matrices $\Psi_1 = \BE S(X; q_0) S^T(X; q_0)$ and $\Psi_{1W}$ are positive definite. Then
\ban
ARE (\hat{q}_{nW}, \hat{q}_{n}) \geq
\frac{ \lambda_{\min} (\Psi_1) \lambda_{\min}^2 (\Psi_{2W})}
{W_{\max} \lambda_{\max} (\Psi_{1W}) \lambda_{\max}^2 (\Psi_2)}.
\ean
Consequently, if $W_{\max} / \lambda_{\min}^2 (\Psi_{2W}) < \lambda_{\min} (\Psi_1) /(\lambda_{\max} (\Psi_{1W}) \lambda_{\max}^2 (\Psi_2))$ then this asymptotic relative efficiency is larger than 1. 
\end{Corollary}

\begin{proof}[Proof of Corollary~\ref{Cor:ARE}]
Using the facts that $det (AB) = det(A) det(B)$ for square matrices $A,B$ and $det (A^{-1}) = 1/det(A)$ for non-singular $A$, we write
\ban
\frac{det (V_1)}{det (V_W)} &= det( \Psi_2^{-1} \Psi_1 \Psi_2^{-1} )
det( \Psi_{2 W} \Psi_{1 W}^{-1} \Psi_{2 W} )\\
&= det (\Psi_1) det (\Psi_{1 W}^{-1}) [ det(\Psi_2^{-1}) det(\Psi_{2 W}) ]^2.
\ean
The result follows, using the facts that $\det(A) \geq \lambda_{\min} (A)$ and $\det(A^{-1}) \geq 1/\lambda_{\max} (A)$, and the upper bound on $W$.
\end{proof}


We may wish to further explore the conditions of Corollary~\ref{Cor:ARE}. 
Let us concentrate on the case of where the distribution of $X$ is spherically symmetric. Following \cite{ref:Fangetal90_Book}:

\begin{Definition}
A $p$-dimensional random vector $X$ is said to elliptically distributed if there exist a vector $\mu \in \BR^p$, a positive semi-definite matrix $\Sigma \in \BR^{p \times p}$ and a function 
$\phi: (0, \infty) \rightarrow \BR$ such that the characteristic function of $X$ is $\exp \{ i t^{T} \mu \} \phi (\bft^T \Sigma \bft)$ for $\bft \in \BR^p$.
\end{Definition}

There are several alternative formulations, see the above reference and citations within it for details.

\begin{Corollary}\label{Cor:Elliptic_ARE}
Assume that $\BF \equiv \BF_X$ is an elliptically symmetric distribution, with location parameter $\mu = q_{0}$ and the covariance matrix $\Sigma$ satisfies the following conditions:
\begin{enumerate}


\item $\exp \left\{ - \frac{Tr^2 (\Sigma)}{ 256 \lambda_{\max}^2 (\Sigma)} \right\} = 
o\left( \min\left( \frac{\lambda_{\max}(\Sigma) }{Tr( \Sigma)},
\frac{\lambda_{\min}(\Sigma)}{\lambda_{\max}(\Sigma)} \right) \right)$,

\item $\lambda_{\max}(\Sigma) = o( Tr(\Sigma))$.
\end{enumerate}

Also assume that the weight function is (a) affine invariant, i.e. $ W (Ax + b, A\BF + b) = W (x, \BF)$ for $A \in \BR^{p \times p}, b \in \BR^p$, and (b) satisfies the following:
\begin{equation}\label{eqn:CorElliptic_ARE_eqn}
\frac{W_{\max}}{[ \BE | X - q_0|^{-1} W(X, \BF)]^2 } <
\frac{\lambda_{\min} (\Psi_1)}{\lambda_{\max} (\Psi_1) [ \BE | X - q_0|^{-1}]^2]}.
\end{equation}
Then we have $ARE( \hat q_{n,W}, \hat q_n) > 1$.
\end{Corollary}

The conditions 1 and 2 in Corollary~\ref{Cor:Elliptic_ARE} are due to \cite{ref:JASA151658_WangPengLi}, and ensure that as $p \equiv p_n$ diverges with $n$, the eigenvalues of $\Sigma \in \BR^{p \times p}$ are bounded away from 0 and $\infty$. Corollary~\ref{Cor:Elliptic_ARE} can be proved using Corollary~\ref{Cor:ARE}, the fact that for elliptical distributions $\Psi_1$ is non-singular \citep{ref:SPL12765_Taskinenetal}, then using and expanding upon Lemma A.5 in \cite{ref:JASA151658_WangPengLi}. We give the technical details in the supplementary material.

Obtaining affine invariant weight functions is not challenging. Most functions arising in the context of data depths are affine invariant. Corollary~\ref{Cor:Elliptic_ARE} implies that there is a wide frame of distributions and  choices of weight functions where there is a benefit to considering weighted spatial medians. In fact, Corollary~\ref{Cor:Elliptic_ARE} can be established under just 
location invariance condition on weights. We use the stronger affine invariance condition
here only because the subsequent sections use weights based on affine invariant 
data depth functions. Also note that in case \eqref{eqn:CorElliptic_ARE_eqn} does not hold for the original form of a weight function, one can scale all weights by an appropriate constant to reduce the value in the left hand side of \eqref{eqn:CorElliptic_ARE_eqn} to satisfy the condition.

In actual computations, in place of $W (x, \BF)$ we propose using $W (x, \BF_{n})$ where
$ \BF_{n}(z) = n^{-1} \sum_{i = 1}^{n} \prod_{j=1}^{p}
\cI (X_{ij} \leq z_j ), z \in \BR^p
$ is the empirical distribution function. Up to first order asymptotics, the analysis remain unchanged from the above for this modified weight function.

\subsection{Examples of affine invariant weights}

We now illustrate some specific choices of weight functions that are compatible with the 
conditions of Corollary~\ref{Cor:Elliptic_ARE} and the results in the rest of this paper.
 These arise as easy transformations 
of \textit{data-depth functions}. A data depth function 
is defined on $\cX \times \cF$, where $\cF$ is  a fixed set of probability measures.
The main property of a data-depth function is that for every probability measure 
$\BF \in \cF$, there exists a constant 
$\mu_{\BF} \in  \cX$ such that for any $t \in [ 0, 1]$ and $x \in \cX$,
\baq 
D ( \mu_{\BF} ; \BF ) \geq D ( \mu_{\BF} + t ( \bfx - \mu_{\BF} ); \BF ). 
\label{eq:Peripherality}
\eaq
That is, for every fixed $\BF$, the data-depth function achieves a supremum at 
$\mu_{F}$, and is non-decreasing in every direction away from $\mu_{F}$, thus providing 
a  center-outward partial ordering of points in $\cX$. There are generally several 
algebraic and analytic properties assumed for data-depth functions to elicit interesting 
mathematical properties, see for example \cite{ref:DIMACS061_Serfling, ref:AoS00461_ZuoSerfling} for details.

The spherically symmetric case of an elliptical distribution is realized with $\Sigma = \sigma^{2} \BI_{p}$ for some $\sigma^{2} > 0$. We fix the notation $Z = \Sigma^{-1/2} (X - \mu)$, and let 
$Z \sim \BF_{Z}$. Note that $\BF_{Z}$ is a spherically symmetric distribution and hence
depends only on $|z|$, and $\BE Z = \mathbf{0}_{p} \in \BR^{p}$ and $\BV Z = \BI_{p}$. Taking affine invariant data depth functions as weights ensures that $W(X, \BF) = W(Z, \BF_Z)$. It is now easy to show that results in this paper are valid for the weight functions (with appropriate scaling to satisfy \eqref{eqn:CorElliptic_ARE_eqn})
($i$) $W_{HSD} (X) = \BF_{Z_1} (| Z |)$ derived from the \textit{half-space depth},  
($ii$) $W_{MhD} (X) = |Z|^{2}/(1 + | Z |^2)$ derived from the \textit{Mahalanobis depth}, 
and ($iii$) $W_{PD} (X) = |Z|/(1+ | Z |/MAD(Z_1))$, where $MAD$ stands for median 
absolute deviation, derived from the \textit{projection depth}. We omit the technical 
details. These three weight functions give a center-inward partial ordering, thus essentially 
quantifying \textit{peripherality} instead of \textit{depth}.
Note however, that our results below are of much more general form, and 
these three special choices of weights only serve as important illustrative examples 
to achieve desirable robustness and efficiency balance in data analysis.

\section{A robust measure of dispersion} 
\label{Sec:WSDispersion1}
From this section onwards, we assume that $\cX = \BR^{p}$, that is, the support of the 
random variable under study is the $p$-dimensional Euclidean plane, and the data 
$X_{1}, \ldots, X_{n}$ are independent and identically distributed from an 
elliptical distribution $\BF$ with parameters $\mu$ and $\Sigma$.  We also 
assume that $X_{1}$ is absolutely continuous, with $\BP [ | X_{1}| = 0] = 0$, 
and that $\Sigma$ is positive definite. This eliminates technicalities arising from rank deficient cases, and makes the weight functions affine invariant.
Thus, we essentially restrict the rest of this paper to the same  framework as in Corollary~\ref{Cor:Elliptic_ARE}. 
Most of the results below generalize to the case where 
$\cX$ is a separable Hilbert space, however additional technicalities are involved, as in 
\cite{ref:AoS112852_Balietal_PCA_Functional_Robust}, and will be considered 
in a future project. 

\subsection{The Weighted Sign Covariance Matrix}
Consider the spectral decomposition of $\Sigma$ given by
 $\Sigma = \Gamma\Lambda\Gamma^T$, where $\Gamma$ is an orthogonal matrix 
 and $\Lambda$ is diagonal with positive diagonal elements $\lambda_1 \geq \ldots \geq 
 \lambda_p$.
Also denote the $i$-th eigenvector of $\Sigma$ by 
$\bfgamma_{i} = (\gamma_{i, 1}, \ldots, \gamma_{i, p})^T$ for $1 \leq i \leq p$. Thus, the 
$i$-th column of $\Gamma$ is $\bfgamma_{i}$.
  In the rest of this paper we use the notation 
 $\Sigma^{-1/2} = \Lambda^{-1/2} \Gamma^{T}$, and 
 hence $Z = \Lambda^{-1/2} \Gamma^T (X - \mu)$. 
Recall from Section~\ref{Sec:WSQuantiles} that we use the notation $\BF_{Z}$ for 
the distribution of $Z$, and that $\BF_{Z}$ is a spherically symmetric distribution 
and hence depends only on $|z|$.
Additionally, to simplify notations, for any random variable $X \sim \BF$, we occasionally 
use the abbreviated notation $W (X) \equiv W (X, \BF)$. Note that $W (X)$ is a 
\textit{random weight}, and takes the same value as $W (Z, \BF_{Z})$.  
Also, as noted in  Corollary~\ref{Cor:Elliptic_ARE}, 
 $W (Z)$ is a function of $|Z|$ only. 
 We additionally assume that $\BE W^{2} (X) < \infty$.
 
It is convenient to write 
\ban
X = \mu + R \Gamma \Lambda^{1/2}  U,
\ean
where $U$  is a random variable uniformly distributed on the unit sphere 
${\cS}_{0; 1} = \{ x \in \cX: | x | = 1 \}$ and $R$ is another random variable independent 
of $U$ satisfying $\BE R^{2} = p$. Note that  $Z = R U$, and $|Z| = R$, $Z/|Z| = U$. 
Then we have
\ban 
S (X; \mu) & = |X - \mu|^{-1} (X - \mu) 
= | \Lambda^{1/2} R U|^{-1} R \Gamma  \Lambda^{1/2}  U\\
& =  | \Lambda^{1/2} U|^{-1} \Gamma \Lambda^{1/2}  U 
=  | \Lambda^{1/2} Z|^{-1} \Gamma \Lambda^{1/2}  Z.
\ean

As a robust surrogate for $X - \mu$, we consider the  following  random variable 
\baq
\tilde{X} = W (X, \BF) S(X; \mu) \equiv W (Z, \BF_{Z})  
| \Lambda^{1/2} Z|^{-1} \Gamma \Lambda^{1/2}  Z. 
\label{eq:tildeX}
\eaq
In samples, the equivalent for $\tilde{X} $ is 
$\hat{\tilde{X}} = W (X, \BF_{n}) S(X; \hat{\mu})$ for a suitable location estimator 
$\hat{\mu}$, for example, the weighted spatial median. 
We fix the notation $\BS (X; \mu) =  S (X; \mu) S( X; \mu)^T$, and define the following dispersion  parameter:
\baq
\tilde{\Sigma} = \BE \tilde{X} \tilde{X}^{T}
=  \BE W^{2} (X, \BF) \BS (X; \mu). 
\label{eq:tildeSigma}
\eaq
In the following Theorem, we establish that the eigenvectors of ${\Sigma}$ and $\tilde{\Sigma}$ are identical, although their eigenvalues may be different.
  
\begin{Theorem}
\label{Thm:WSVariance}
Under the conditions listed above, we have
$\tilde{\Sigma} = \Gamma \tilde{\Lambda} \Gamma^{T}$, where 
$\tilde{\Lambda} = \Lambda^{1/2} \BE W^{2} (X) 
	\BE [UU^{T}/(U^{T}\Lambda U)] \Lambda^{1/2}$
is a diagonal matrix.  Thus, the eigenvectors of ${\Sigma}$ and 
 $\tilde{\Sigma}$ are identical.
\end{Theorem}

\begin{proof}[Proof of Theorem~\ref{Thm:WSVariance}]

Fix any index $i \in \{1, \ldots, p \}$. Consider the vector $\tilde{U}$ such that 
\ban 
\tilde{U}_{j} = \left\{ \begin{array}{ll}
U_{j} & \text{ if } j \ne i, \\
- U_{i} & \text{ if } j = i. 
\end{array}
\right. 
\ean
Then $\tilde{U}$ and $U$ have the same distribution, and note that 
$U^{T} \Lambda U = \tilde{U}^{T} \Lambda \tilde{U}$ almost surely. Consequently, 
for any $j \ne i$ we have
\ban
\BE {\frac{U_{i} U_{j}}{U^{T} \Lambda U}}  
& = \BE {\frac{\tilde{U}_{i} \tilde{U}_{j}}{\tilde{U}^{T} \Lambda \tilde{U}}} 
= - \BE {\frac{{U}_{i} {U}_{j}}{U^{T} \Lambda U}}.
\ean
Consequently, 
$\BE S (X; \mu) S(X; \mu)^{T} = \Gamma \Lambda_{S} \Gamma^{T}$, 
as established in Theorem~1 of \citet{ref:SPL12765_Taskinenetal}.

Also, since the weight $W(X)$ is a function of $|Z| = R$, we have that $W(X)$ is 
independent of $S (X; \mu)$. Consequently, we have  
  \ban
  \tilde{\Sigma} & = \BE \tilde{X} \tilde{X}^{T} 
 = \BE W^{2} (X) S (X; \mu) S(X; \mu)^{T} \\
& =  \BE W^{2} (X)  \BE S (X; \mu) S(X; \mu)^{T} 
= \Gamma \Lambda_{W} \Gamma^{T}, 
\ean
where $\Lambda_{W}$ is a diagonal matrix.
\end{proof}

\subsection{Sample version of $\tilde \Sigma$}
We now discuss the properties of the sample version of $\tilde \Sigma$, say $\widehat{\tilde{\Sigma}}$ computed from $\bfX$. In practice, we cannot obtain $W (x) \equiv W (x, \BF)$, and consequently use  $W (x, \BF_{n})$ instead. We assume the following conditions:

\begin{description}
\item[(B1)]{\it Bounded weights}: The weights $W (\cdot, \cdot)$ are bounded functions. 

\item[(B2)]{\it Uniform convergence}:
\ban 
\sup_{x \in \cX} | W (x, \BF_{n}) - W (x, \BF) | \rightarrow 0
\ean 
almost surely as $n \rightarrow \infty $.

\item[(B3)]
{\it Smoothness under perturbation:} For all $\BF \in \cF$, there exists a $\delta > 0$, possibly depending on $\BF$, such that 
for any $\epsilon \in (0, \delta)$ 
\ban 
\sup_{x \in \cX} \Bigl| W \bigl( x, \BF \bigr) - 
W \bigl( x, (1 - \epsilon)\BF + \epsilon \delta_{x} \bigr) \Bigr| \leq 
\epsilon.
\ean 
\end{description}
In the above, $\delta_{x}$ denotes point mass at $x$. 
These properties are easily satisfied  for  weight functions  derived from standard depth functions, for example, $W_{HSD} (\cdot)$, $W_{MhD} (\cdot)$ and $W_{PD} (\cdot)$ discussed earlier.

The following result allows us to use the empirical, plug-in weights and an estimated location parameter in the weighted dispersion estimator. A natural choice for the location parameter estimator is the solution to $\sum_{i = 1}^{n} \tilde{X}_{i} = 0$, which is the same as the sample version of the weighted spatial median discussed in Section~\ref{Sec:WSQuantiles}.

\begin{Lemma} \label{Lemma:lemma1}
Assume that $\BE | X - \mu |^{-4} < \infty$. Also assume that 
we have a location estimator  $\hat{\mu}_n$ satisfying 
$\BE |\hat{\mu}_n  - \mu |^{4} = O (n^{-2}) $. Then
\ban
\frac{1}{n} \sum_{i=1}^{n} W_{n}^{2} (X_{i}, \BF_{n}) \BS (X_{i}; \hat{\mu}_{n})  
= \frac{1}{n} \sum_{i=1}^{n} W^{2} (X_{i}, \BF) \BS (X_{i}; {\mu})
+ R_n,
\ean
where for any $c \in \BR^{p}$ such that for $| c | = 1$, we have 
$\BE  c^{T} R_{n} c = o (n^{-1})$.
\end{Lemma}

\begin{proof}[Proof of Lemma~\ref{Lemma:lemma1}]
Note that the moment condition above is satisfied by default since $X$ has an elliptical distribution \citep{durre14}. The proof is mostly algebra, and we provide a sketch of the main arguments.
We have
\ban 
& \frac{1}{n} \sum_{i=1}^{n} W_{n}^{2} (X_{i}, \BF_{n}) \BS (X_{i}; \hat{\mu}_{n}) \\
= & \frac{1}{n} \sum_{i=1}^{n} W^{2} (X_{i}, \BF) \BS (X_{i}; {\mu})
+ \frac{1}{n} \sum_{i=1}^{n}  
\bigl\{ W_{n}^{2} (X_{i}, \BF_{n})  -   W^{2} (X_{i}, \BF) \bigr\} \BS (X_{i}; {\mu}) \\
& + \frac{1}{n} \sum_{i=1}^{n} W^{2} (X_{i}, \BF) \bigl\{ 
\BS (X_{i}; \hat{\mu}_{n}) - \BS (X_{i}; {\mu}) \bigr\} \\
& + \frac{1}{n} \sum_{i=1}^{n}  
\bigl\{ W_{n}^{2} (X_{i}, \BF_{n})  -   W^{2} (X_{i}, \BF) \bigr\} 
\bigl\{ \BS (X_{i}; \hat{\mu}_{n}) - \BS (X_{i}; {\mu}) \bigr\}\\
& = \frac{1}{n} \sum_{i=1}^{n} W^{2} (X_{i}, \BF) \BS (X_{i}; {\mu})
+ T_{2} + T_{3} + T_{4}. 
\ean

Using the stated technical conditions, we can now show that 
$\BE  c^{T} T_{j} c = o (n^{-1})$ for $j = 2, 3, 4$. For illutration, we present
the case for $T_{2}$ below. 

Notice that the $(j, k)$-th element of $T_{2}$ is given by 
\ban 
n^{-1} \sum_{i=1}^{n}  | X_{i} - \mu |^{-2}
\bigl\{ W_{n}^{2} (X_{i}, \BF_{n})  -   W^{2} (X_{i}, \BF) \bigr\} 
(X_{i, j} - \mu_{j}) (X_{i, k} - \mu_{k}), 
\ean
and hence 
\ban 
& c^{T} T_{2} c  = \sum_{j, k} c_{j} c_{k}T_{2, j, k} \\
& = n^{-1} \sum_{i=1}^{n}  | X_{i} - \mu |^{-2}
\bigl\{ W_{n}^{2} (X_{i}, \BF_{n})  -   W^{2} (X_{i}, \BF) \bigr\} 
\bigl( \sum_{j} c_{j} (X_{i, j} - \mu_{j}) \bigr)^{2} \\
& \leq 
M n^{-1} \sum_{i=1}^{n}  | X_{i} - \mu |^{-2}
\bigl\{ | W_{n} (X_{i}, \BF_{n})  -   W (X_{i}, \BF) | \bigr\} 
\bigl( c^{T} (X_{i} - \mu) \bigr)^{2} \\
& \leq 
M n^{-1} \sum_{i=1}^{n}  | X_{i} - \mu |^{-2}
\bigl\{ | W_{n} (X_{i}, \BF_{n})  -   W_{n} (X_{i}, \BF_{n, -i})  | \bigr\} 
\bigl( c^{T} (X_{i} - \mu) \bigr)^{2} 
\\ & \hspace{1cm}
+ M n^{-1} \sum_{i=1}^{n}  | X_{i} - \mu |^{-2}
\bigl\{ | W_{n} (X_{i}, \BF_{n, -i})  -   W (X_{i}, \BF) | \bigr\} 
\bigl( c^{T} (X_{i} - \mu) \bigr)^{2} 
\\
& =  M n^{-1} \sum_{i=1}^{n}  T_{2 1 i}  + M n^{-1} \sum_{i=1}^{n}   T_{2 2 i} \\
& = T_{21} + T_{22}.
\ean
Let $H (X_{i}) = | X_{i} - \mu |^{-2} \bigl( c^{T} (X_{i} - \mu) \bigr)^{2}$, and notice 
that $H (X_{i}) \leq 1$ almost surely for $|c| = 1$.  
Now notice that conditional on $X_{i}$ except for a null set $A_{i}$ (possibly depending 
on $X_{i}$) we have  
$T_{2 1 i} \leq n^{-1}  H (X_{i}) $. Thus, except for a null set 
$A_{1} \cap \ldots \cap A_{n}$, $T_{21} \leq M n^{-2}  H (X_{i})$ and the conclusion 
follows for this part.

The argument for $T_{22}$ follows a similar argument.
\end{proof}

Let $vec(\BS (X; \mu))$ be the vectorized version of  $\BS (X; \mu)$.
We are now in a position to state the result for consistency of the sample 
version of $\tilde{\Sigma}$,

\begin{Theorem} \label{Theorem:CLT1}
Assume the conditions of Lemma~\ref{Lemma:lemma1}. Then
\ban
& n^{1/2} \sum_{i = 1}^{n} \Bigl( 
W_{n}^{2} (X_{i}, \BF_{n}) vec(\BS (X_{i}; \hat{\mu}_{n})) 
- \BE W^{2} (X_{i}) vec(\BS (X_{i}; \mu)) \Bigr)\\
& \hspace{0.5cm} \draro
N_{p^2} \bigl( 0, V_{W} \bigr),
\ean
where $V_{W} = \BV [W^{2} (X, \BF) vec(\BS (X; {\mu})) ] $.
\end{Theorem}

The asymptotic normality  follows from our assumptions and as a direct consequence 
of Lemma~\ref{Lemma:lemma1}. Incidentally, an expression for $V_{W}$
can be explicitly obtained in terms of $\Gamma$, $\Lambda$ and $\BF$, but is 
algebraic in nature. We present it in the supplementary material.

We now use Theorem~\ref{Theorem:CLT1} to obtain consistency results for 
the eigenvectors obtained from $\widehat{\tilde{\Sigma}} = n^{-1} 
\sum_{i=1}^{n} W^{2} (X_{i}, \BF) \BS (X_{i}; \hat{\mu}_{n})$. 
Suppose that  $\tilde{\Lambda}_{1} > \tilde{\Lambda}_{2} > \ldots > \tilde{\Lambda}_{p}$
are the eigenvalues of $\tilde{\Sigma}$, which we assume are all distinct values. 

\begin{Theorem} \label{Theorem:Eigen1}
Suppose the  spectral decomposition of $\widehat{\tilde{\Sigma}}$ is given 
by $\widehat{\tilde{\Sigma}} = \widehat{\Gamma} \widehat{\tilde{\Lambda}}
\widehat{\Gamma}^T $. Then the matrix of centered and scaled eigenvectors 
$G_{n} = n^{1/2} (\widehat{\Gamma} - \Gamma) $ 
and the vector of centered and scaled eigenvalues 
$L_{n} = n^{1/2} (\widehat{\tilde{\Lambda}} - {\tilde{\Lambda}}) $ have 
independent distributions. The distribution of the random variable $vec~(G_{n})$ converges 
weakly to a $p^2$-variate normal distribution with mean ${\bf 0}_{p^2}$ and the
variance matrix whose $(i, j)$-th block of $p \times p$ entries are given by
\baq
& \sum_{k=1, {k \neq i}}^{p} 
\Bigl[ \tilde{\Lambda}_{i} -  \tilde{\Lambda}_{k} \Bigr]^{-2}
\BE \Bigl[ 
W^4 (Z, \BF_{Z}) \big( 
\BS_{i, k} ({\Lambda}^{1/2} Z; {\bf 0})
\big)^2 
\Bigr]
\bfgamma_k \bfgamma_k^T, 
\text{ if } i = j, \label{equation:DevEq} \\
& - 
\Bigl[ \tilde{\Lambda}_{i} -  \tilde{\Lambda}_{j} \Bigr]^{-2}
\BE \Bigl[ 
W^4 (Z, \BF_{Z}) \big( 
\BS_{i, j} ({\Lambda}^{1/2} Z; {\bf 0})
\big)^2 
\Bigr]
\bfgamma_i \bfgamma_j^T; \text{ if } i \neq j.
\eaq
The distribution of  $L_{n}$ converges 
weakly to a $p$-dimensional normal distribution with mean ${\bf 0}_{p}$ and the
variance-covariance matrix whose $(i, j)$-the element is
\ban
& \BE \Bigl[ 
W^4 (Z, \BF_{Z}) \big( 
\BS_{i, i} ({\Lambda}^{1/2} Z; {\bf 0})
\big)^2 
\Bigr]
- \tilde{\Lambda}_{i}^2, \text{ if } i = j, \\
& \BE \Bigl[ 
W^4 (Z, \BF_{Z}) \big( 
\BS_{i, j} ({\Lambda}^{1/2} Z; {\bf 0})
\big)^2 
\Bigr]
- \tilde{\Lambda}_{i} \tilde{\Lambda}_{j}, \text{ if } i \neq j.
\ean
\end{Theorem}

The proof of this result follows from using Theorem~\ref{Theorem:CLT1} and using 
techniques similar to a corresponding result in \citet{ref:SPL12765_Taskinenetal}. We omit the algebraic details here and put it in the supplementary material.

Recall that the asymptotic variance of the $i$-th eigenvector of the 
\textit{sample covariance matrix}, say $\hat{\bfgamma}_i$ is \citep{ref:AndersonBook09}:
\begin{equation} \label{equation:covevEq}
A\BV( \sqrt n\hat \bfgamma_{i}) = 
\sum_{k=1; k \neq i}^p 
\frac{\lambda_i \lambda_k}{(\lambda_i - \lambda_k)^2} \bfgamma_k \bfgamma_k^T; 
\quad 1 \leq i \leq p.
\end{equation}
Suppose $\widehat{\tilde{\bfgamma}}_i$ is the $i$-th eigenvector of 
$\widehat{\tilde{\Sigma}}$, whose asymptotic behavior is presented above in 
Theorem~\ref{Theorem:Eigen1}.

This leads to the following useful result:
\begin{Corollary}
The asymptotic relative efficiency of $\widehat{\tilde{\bfgamma}}_i$, 
relative to $\hat{\bfgamma}_i$, is given by 
\ban
& ARE (\widehat{\tilde{\bfgamma}}_i, \hat{\bfgamma}_i; \BF) \\
& =  
\Bigl[ \sum_{k=1; k \neq i}^p \frac{\lambda_i \lambda_k}{(\lambda_i - \lambda_k)^2} \Bigr]
 \Bigl[
\sum_{k=1, {k \neq i}}^{p} 
\biggl[ \tilde{\Lambda}_{i} -  \tilde{\Lambda}_{k} \biggr]^{-2}
\BE \biggl[ 
W^4 (Z, \BF_{Z}) \big( 
\BS_{i, k} ({\Lambda}^{1/2} Z; {\bf 0})
\big)^2 
\biggr]  
\Bigr]^{-1}.
\ean
\end{Corollary}

The proof of this Corollary is immediate.

\section{An affine equivariant robust measure of dispersion}
\label{Sec:WSDispersion2}
A desirable invariance property of any dispersion parameter $T_{X}$ corresponding to 
a random variable $X$ is that under affine transformation $Y = AX + b$ the dispersion 
parameter scales to $T_{Y} = A T_{X} A^{T}$. It is clear that $\tilde{\Sigma}$ does 
not possess this property, since it remains unchanged for $X$ and $Y = c X$ for any 
scalar $c > 0$. 

We follow the general framework of M-estimation with data-dependent weights 
\citep{ref:HuberBook81} to construct an affine equivariant counterpart of the 
$\tilde{\Sigma}$. 
Specifically, we implicitly define
\begin{equation} \label{eqn:ADCM}
\Sigma_{*} = \frac{p}{ \BV W (X) } 
\BE \left[ \frac{W^{2}(X) (X - \mu) (X - \mu)^T}
{(X - \mu)^T \Sigma_{*}^{-1}(X - \mu)} \right].
\end{equation}

To ensure existence and uniqueness of $\Sigma_{*}$, consider the class of 
dispersion parameters $\Sigma_M$ that are obtained as solutions of the following equation:
\begin{equation}
\BE \left[ u( | Z_{M} | )  \frac{Z_{M} Z_{M}^T}{| Z_{M} |^2}  - v( | Z_{M} | ) \BI_p \right] = 0,
\end{equation}
with $Z_M = \Sigma_M^{-1/2} (X - \mu)$. Under the following assumptions on the scalar valued functions $u$ and $v$, the above equation produces a unique solution \citep{ref:HuberBook81}:

\begin{description}
\item[(C1)] The function $u(r)/r^2$ is monotone decreasing, 
and $u(r)>0$ for $r > 0$;

\item[(C2)]  The function $v(r)$ is monotone decreasing, and 
$v(r) > 0$ for $r > 0$;

\item[(C3)] Both $u(r)$ and $v(r)$ are bounded and continuous,

\item[(C4)] $u(0) / v(0) < p$,

\item[(C5)] For any hyperplane in the sample space $\mathcal X$, 
(i) $P(H) = \BE \{ \cI_{\{X \in H \}} \} < 1 - p v(\infty) / u(\infty)$ and 
(ii) $P(H) \leq 1/p$.
\end{description}

Putting things into context, in our case we have 
$v (\cdot) = p^{-1}\BV W (X)$, 
$u (\cdot) = W^{2}(X)$. 
We proceed to verify the other conditions for the weight functions 
$W_{HSD} (\cdot)$, 
$W_{MhD} (\cdot)$ and $W_{PD} (\cdot)$ discussed earlier.

It is easy to verify that the resulting $u (\cdot)$ from the above choices 
satisfy (C1) and (C3). 
Note that $v (\cdot)$ is a finite positive constant, 
and (C2) and (C3) are also easily satisfied. 
Since $u (0) = 0$ in all the above cases, (C4) is also easy to check. Since $X$ is 
absolutely continuous, (C5) holds trivially.

To compute the sample version of $\Sigma_{*}$, we solve (\ref{eqn:ADCM}) 
iteratively by obtaining a sequence of positive definite matrices 
$\hat{\Sigma}^{(k)}_{*}$ until convergence. Thus, using the location 
estimator $\hat{\mu}_{n}$, we may iterate
\ban
\hat{\Sigma}^{(k+1)}_{*}  = 
\frac{p}{ \BV W (X) } 
\BE \left[ \frac{W^{2}(X) (X - \hat{\mu}_{n}) (X - \hat{\mu}_{n})^T}
{(X - \hat{\mu}_{n})^T (\hat{\Sigma}^{(k)}_{*})^{-1}(X - \hat{\mu}_{n})} \right].
\ean

The asymptotic properties of $\hat{\Sigma}_{*}$ can be obtained using methods similar 
to those of Section~\ref{Sec:WSDispersion1}, and techniques presented in 
\cite{ref:Biometrika00603_CrouxHaesbroeck} and elsewhere. We state the following result and omit its proof.

\begin{Theorem}
\label{Thm:Eigen2}
The asymptotic covariance matrix of an eigenvector of the sample 
affine equivariant scatter functional $\hat{\Sigma}_{*}$ is given by
\ban 
V_{12}
\sum_{k=1, k \neq i}^p \frac{\lambda_i \lambda_k}{\lambda_i - \lambda_k} 
\bfgamma_i \bfgamma_k^T,
\ean
where $V_{12}$  is the asymptotic variance of an off-diagonal element of 
$\hat{\Sigma}_{*}$ when the underlying distribution is $\BF_{Z}$. 
It follows that if $\hat{\bfgamma}_{*, i}$ is the $i$-th eigenvector of $\hat{\Sigma}_{*}$,
\begin{equation}
ARE (\hat{\bfgamma}_{*, i}, \hat\bfgamma_{i}; \BF) = V_{12}^{-1} = 
\frac{\left[ \BE ( p u (| Z |)  + u'( | Z |) | Z | ) \right]^2}
{p^2 (p+2)^2 \BE (u (| Z |)^2 \BE (\BS_{12} (Z; {\bf 0}))^2}.
\end{equation}
\end{Theorem}

\section{Robust estimation of eigenvalues, and a plug-in estimator of $\Sigma$}
\label{Sec:Eigen}
As seen in Theorem~\ref{Thm:WSVariance}, eigenvalues of the $\tilde{\Sigma}$ 
are not same as the population eigenvalues. In this section, we discuss on robust 
estimation of $\lambda_{i}$'s using $\tilde{\Sigma}$. Assume the data 
is centered, the robust estimator from Section~\ref{Sec:WSQuantiles} suffices. 
We start by computing 
the sample version $\widehat{\tilde{\Sigma}}$ and its spectral decomposition:
\ban 
\widehat{\tilde{\Sigma}} = \widehat{{\Gamma}} \widehat{\tilde{\Lambda}}
\widehat{{\Gamma}}^{T}.
\ean 
We then use the following steps:

\begin{enumerate}
\item Randomly divide the sample indices $\{1,2, \ldots, n\}$ into $k$ disjoint groups $\{G_1,\ldots, G_k \}$ of size $\lfloor n/k \rfloor$ each.

\item Transform the data matrix: 
$\bfS = \widehat{{\Gamma}}^{T} \bfX$.

\item Calculate coordinate-wise variances for each group of indices $G_j$:
\ban
\lambda_{i, j}^{\dagger} = \frac{1}{|G_j|} \sum_{l \in G_j} \bigl( 
S_{l i} - \bar{S}_{G_{j}, i} \bigr)^2; \quad i = 1, \ldots, p; \  j = 1, \ldots, k.
\ean
where $\bar{\bfS}_{G_j} = (\bar{S}_{G_j, 1}, \ldots, \bar{S}_{G_j,p})^T$ is the vector of 
column-wise means of $\bfS_{G_j}$, the submatrix of $\bfS$ with row indices in $G_j$.

\item Obtain estimates of eigenvalues by taking coordinate-wise medians of these variances:
\ban
{\lambda}^{\dagger}_{i} = \text{median} (\lambda_{i,1}^{\dagger}, 
\ldots , \lambda_{i,k}^{\dagger} ); \quad 
i = 1, \ldots, p.
\ean
\end{enumerate}

We collect ${\lambda}^{\dagger}_{i}$, $ i =1, \ldots, p$ in the diagonal matrix 
${\Lambda}^{\dagger} = \diag ({\lambda}^{\dagger}_{1}, \ldots, 
{\lambda}^{\dagger}_{p})$.
The number of subgroups used to calculate this median-of-small-variances estimator can 
be determined following \cite{ref:Bernoulli152308_Minsker_Median_Banach}. 
There can be other ways of estimating 
the eigenvalues of $\Sigma$ using $\bfS$ also, we will pursue such methods elsewhere.
We construct a consistent 
plug-in estimator of the population covariance matrix $\Sigma$ as 
\ban 
{\Sigma}^{\dagger}
= \widehat{{\Gamma}} {\Lambda}^{\dagger} \widehat{{\Gamma}}^{T}.
\ean
Let $| A |_{F}$ denote the Frobenius norm of a matrix $A$, in other words,
$| A |_{F} = (\text{trace} A^{T} A)^{1/2}$.
The following result establishes that this is a consistent estimator of $\Sigma$:

\begin{Theorem}\label{Thm:pluginSigma}
Suppose that as $n \rightarrow \infty$, $k \rightarrow \infty$ and 
$n/k \rightarrow \infty$.
Then we have
\ban
\| {\Sigma}^{\dagger} - \Sigma \|_F \stackrel{P}{\rightarrow} 0.
\ean
\end{Theorem}

\begin{proof}[Proof of Theorem \ref{Thm:pluginSigma}]
This proof has many algebraic steps, and we sketch the main arguments below.

Suppose $\hat{A} = \widehat{{\Gamma}}^{T}  \Sigma \widehat{{\Gamma}}$.  Owing to the fact that the Frobenius norm is invariant under rotations and that $p$ is finite and fixed, it suffices to show that the off-diagonal elements of $\hat{A}$ converge in probability to zero, and that the difference between the $i$-th diagonal element of $\hat{A}$ and ${\lambda}^{\dagger}_{i}$ converges to zero for any $i = 1, \ldots, p$.

Now notice that from Theorem~\ref{Theorem:Eigen1} we have that $\widehat{{\Gamma}} 
= {{\Gamma}} + R_{n1}$, where the $(i, j)$-th element of the remainder 
$R_{n1, i, j}$ satisfies $\BE  R_{n1, i, j}^{2} = O(n^{-1})$. We can show, using standard algebra, that 
\ban
\hat{A} = \Lambda + R_{n2}, 
\ean
where the $(i, j)$-th element of the remainder 
$R_{n2, i, j}$ satisfies $\BE  R_{n2, i, j}^{2} = O(n^{-1})$. 
This follows immediately from above, the fact that $p$ is finite and fixed, and all 
elements of $\Lambda$ are constants. This immediately establishes the case for the 
off-diagonal elements. 

For the diagonal elements, notice that since $k \rightarrow \infty$, each 
coordinate-wise variance $\lambda_{i, j}^{\dagger}$ for each group of indices $G_j$ 
is a consistent estimator of $\lambda_{i}$. The result follows.

\end{proof}

\section{Influence Functions of Dispersion Measures}
\label{Sec:RE_Dispersion}

We retain the framework adopted in Section~\ref{Sec:WSDispersion1}, and discuss in this 
section the robustness and efficiency properties associated with $\tilde{\Sigma}$ 
and $\Sigma_{*}$, and principal components derived therefrom. 
We do not discuss $\Sigma^{\dagger}$ here, since  the 
properties of that approach follow from those of  $\tilde{\Sigma}$. 
We additionally assume that the eigenvalues of $\Sigma$ are 
distinct, and given by  $\lambda_1 > \lambda_2 > \ldots > \lambda_p$, to avoid 
several additional technical conditions for the theoretical results to follow. 
The case where the eigenvalues of $\Sigma$ can have multiplicity greater than one 
requires no additional conceptual development, but does require considerable algebraic 
manipulations.

For studying 
the robustness aspect, we first present some results relating to influence functions 
in the current context. 
 Influence functions quantify how much influence a sample point, especially 
an infinitesimal contamination, has on any functional of a probability 
distribution \citep{ref:HampelBook86}. Given any probability distribution 
$\BH \in \cM$, the influence function of any point $x_0 \in \mathcal{X}$ for 
some functional $T(\BH)$ on the distribution is defined as
\ban
IF(x_0; T,\BH) =
\lim_{\epsilon \rightarrow 0} \frac{1}{\epsilon} (T(\BH_\epsilon) - T(\BH)),
\ean
where 
$\BH_\epsilon = (1-\epsilon)\BH + \epsilon \delta_{x_0}$; $\delta_{x_0}$ being 
the distribution with point mass at $x_0$. When $T(\BH) = E_\BH f$ for some 
$\BH$-integrable function $f$, $IF(x_0; T,\BH) = f(x_0) - T(\BH)$.

It now follows that
\ban
IF (x_0; \tilde{\Sigma}, \BF) =
W^{2} (x_{0}) \BS(x_0; \mu) - \tilde \Sigma.
\ean
Recall that $\tilde{\lambda}_{1} > \tilde{\lambda}_{2} > \ldots > \tilde{\lambda}_{p}$
are the eigenvalues of $\tilde{\Sigma}$, which we assume are all distinct values. 

\begin{Proposition}\label{Thm:IF}
The influence function of $\bfgamma_{i}$ as follows: 
\ban
IF(x_0; \bfgamma_{i}, \BF)  & = 
\sum_{k = 1; k \neq i}^p \frac{1}{\tilde{\lambda}_{i} - \tilde{\lambda}_{k}} 
\left\{ \bfgamma^T_k IF(x_0; \tilde \Sigma,\BF)  \bfgamma_{i} \right\} 
\bfgamma_k \notag \\
& =  \sum_{k=1; k \neq i}^p \frac{1}{\tilde{\lambda}_{i} - \tilde{\lambda}_{k}}
W^{2} (x_{0}) \left\{  \bfgamma^T_{k} \BS (x_0; \mu) \bfgamma_{i} \right\} 
\bfgamma_{k}.
\ean
\end{Proposition}

The proof of Proposition~\ref{Thm:IF} follows from \cite{ref:JRSSB79217_Sibson} and 
\cite{ref:Biometrika00603_CrouxHaesbroeck}, we omit the details. 

If the weight function $W (\cdot)$ is a bounded function, as is the case of $W_{HSD}$, 
$W_{MhD}$, and $W_{PD}$, 
the influence function given in  Proposition~\ref{Thm:IF} is bounded, indicating good 
robustness properties of the principal component analysis.

We now derive the influence function for $\Sigma_{*}$. 
\begin{Proposition}\label{Thm:IF2}
The influence function of $\Sigma_{*}$ is given by
\ban 
IF ( x_{0}, \Sigma_{*}, \BF) = 
\alpha_{\Sigma_{*}} ( | x_{0} |; \BF_{Z} ) 
\BS( x_0; \mu) 
- \beta_{\Sigma_{*}} (| x_{0} |; \BF_{Z} ) \Sigma_{*}.
\ean
\end{Proposition}

\begin{proof}[Proof of Proposition~\ref{Thm:IF2}]
Let $z_{0} = \Lambda^{-1/2} \Gamma^T  (x_0 - \mu) = (z_{0 1}, \ldots, z_{0 p})^T$.
As a first step, since $\Sigma_{*}$ is affine equivariant, we obtain from 
\cite{ref:Biometrika00603_CrouxHaesbroeck} 
that 
\ban 
IF ( x_{0}, \Sigma_{*}, \BF) = \Sigma_{*}^{1/2} 
IF ( z_{0}, \Sigma_{*}, \BF_{Z})  \Sigma_{*}^{1/2}.
\ean
From Lemma 1 of \citep{ref:HampelBook86}, page 276, we obtain that there exist 
scalar valued functions $\alpha_{\Sigma_{*}} ( | x_{0} |; \BF_{Z} )$ and 
$\beta_{\Sigma_{*}} ( | x_{0} |; \BF_{Z} ) $ such  that 
\ban
IF ( z_{0}, \Sigma_{*}, \BF_{Z}) = \alpha_{\Sigma_{*}} ( | x_{0} |; \BF_{Z} ) 
\BS( z_0; {\mathbf 0}) 
- \beta_{\Sigma_{*}} (| x_{0} |; \BF_{Z} ) \BI_{p}, 
\ean
consequently we obtain
\ban 
IF ( x_{0}, \Sigma_{*}, \BF) = 
\alpha_{\Sigma_{*}} ( | x_{0} |; \BF_{Z} ) 
\BS( x_0; \mu) 
- \beta_{\Sigma_{*}} (| x_{0} |; \BF_{Z} ) \Sigma_{*}.
\ean
\end{proof}

Suppose ${\lambda}_{* 1} > {\lambda}_{* 2} > \ldots > {\lambda}_{* p}$
are the eigenvalues of ${\Sigma}_{*}$, which we assume are all distinct values. 
Also denote the $i$-th eigenvector of $\Sigma_{*}$ by 
$\bfgamma_{* i} = (\gamma_{* i 1}, \ldots, \gamma_{* i p})^T$ for $1 \leq i \leq p$.

\begin{Proposition}\label{Thm:IF3}
The influence function of $\bfgamma_{* i} $ may be obtained as 
\ban
IF(x_0; \bfgamma_{*i}, \BF)  & = 
\sum_{k = 1; k \neq i}^p \frac{1}{{\lambda}_{*i} - {\lambda}_{*k}} 
\left\{ \bfgamma^T_{* k} IF(x_0; \tilde \Sigma,\BF)  \bfgamma_{* i} \right\} 
\bfgamma_{* k} \notag \\
& =  
\alpha_{\Sigma_{*}} ( | x_{0} |; \BF_{Z} ) \sum_{k=1; k \neq i}^p 
\frac{1}{{\lambda}_{*i} - {\lambda}_{*k}}
\left\{ \bfgamma^T_{* k} \BS (x_0; \mu) \bfgamma_{* i} \right\} 
\bfgamma_{* k}.
\ean
\end{Proposition}

We omit the proof of Proposition~\ref{Thm:IF3}, which follows along similar lines to 
to rest of the computations of this section.
It can be shown that when $W (\cdot)$ is a bounded function, 
$\alpha_{\Sigma_{*}} ( | x_{0} |; \BF_{Z} )$ is also bounded, along the lines of 
\cite{ref:HuberBook81}, which in turn implies that the influence function for a 
principal component based on $\Sigma_{*}$ is also bounded.

\section{Simulation Studies}
\label{Sec:Simulation}

We report a number of numerical simulation studies on several properties relating to 
$\tilde{\Sigma}$ and  $\Sigma_{*}$, and their eigenvalues and eigenvectors, 
on datasets with or without influential points, to illustrate the finite sample 
efficiency and robustness properties of the proposed weighted estimators. We compare 
these proposed estimators with techniques that exists in literature, specifically, 
the Sign Covariance Matrix (SCM) and Tyler's shape matrix \citep{ref:AoS87234_Tyler}.

\subsection{Efficiency of different robust estimators}

We compare the performance of $\tilde{\Sigma}$ and $\Sigma_{*}$ with that of the 
SCM and Tyler's scatter matrix. For this study, we fix the dimension $p = 4$.
We consider six elliptical distributions, 
and from every distribution draw 10000 samples each for sample sizes $n = 20, 50, 100, 
300$ and $500$. All distributions are centered at ${\bf 0}_p$, and have covariance matrix 
$\Sigma = \diag(4, 3, 2, 1)$. 

We use the concept of principal angles 
\citep{ref:LinearAlgebraApplications9281_MiaoBenIsrael} 
to find out error estimates for the 
first eigenvector of a scatter matrix. In our case, the first eigenvector is
\ban
\bfgamma_1 = (1, \overbrace{0, \ldots, 0}^{p - 1})^T.
\ean
We measure the prediction error for an eigenvector estimator (say, $\tilde\bfgamma_{1}$), 
using the smallest angle between the true and predicted vectors, i.e. 
$ \cos^{-1} | \tilde\bfgamma_{1}^T \hat\bfgamma_1 | $. A small absolute value of this 
angle means to better prediction. We repeat this 10,000 times and calculate the 
\textbf{Mean Squared Prediction Angle}:
\ban
MSPA(\hat \bfgamma_{1}) =
\frac{1}{10000} \sum_{m=1}^{10000} \left( \cos^{-1} 
\left|\bfgamma_1^T \tilde\bfgamma^{(m)}_{1} \right| \right)^2.
\ean
where $\tilde\bfgamma^{(m)}_{1}$ is the value of $\tilde\bfgamma_{1}$ in the 
$m$-th replication, $m = 1, \ldots, 10,000$. 
The finite sample efficiency of $\tilde\bfgamma_{1}$ relative to that 
from the sample covariance matrix, i.e. $\hat\bfgamma_{1}$ is obtained as:
\ban
 FSE ( \hat\bfgamma_{1}, \hat\bfgamma_{1}) = 
 \frac{ MSPA(\hat\bfgamma_{0, 1})}{MSPA(\hat\bfgamma_{1})}.
 \ean
 
 \begin{table}[t]
 \centering
\begin{scriptsize}
    \begin{tabular}{c|cc|ccc|ccc}
    \hline
    4-variate $t_5$    & SCM  & Tyler & $\tilde{\Sigma}$-H & $\tilde{\Sigma}$-M & $\tilde{\Sigma}$-P & ${\Sigma}_{*}$-H & ${\Sigma}_{*}$-M & ${\Sigma}_{*}$-P \\ \hline
    $n$=20             & 1.04 & 1.02  & 1.10   & 1.07   & 1.02  & 1.09    & 1.07    & 0.98   \\
    $n$=50             & 1.08 & 1.08  & 1.16   & 1.16   & 1.13  & 1.19    & 1.19    & 1.13   \\
    $n$=100            & 1.31 & 1.31  & 1.42   & 1.38   & 1.36  & 1.46    & 1.44    & 1.36   \\
    $n$=300            & 1.46 & 1.54  & 1.81   & 1.76   & 1.95  & 1.88    & 1.88    & 1.95   \\
    $n$=500            & 1.92 & 1.93  & 2.23   & 2.03   & 2.31  & 2.35    & 2.19    & 2.39   \\ \hline
    4-variate $t_6$     & SCM  & Tyler & $\tilde{\Sigma}$-H & $\tilde{\Sigma}$-M & $\tilde{\Sigma}$-P & ${\Sigma}_{*}$-H & ${\Sigma}_{*}$-M & ${\Sigma}_{*}$-P \\ \hline
    $n$=20             & 1.00 & 1.05  & 1.03   & 1.05   & 1.00  & 1.04    & 1.04    & 0.95   \\
    $n$=50             & 1.03 & 1.01  & 1.13   & 1.12   & 1.11  & 1.19    & 1.17    & 1.10   \\
    $n$=100            & 1.08 & 1.12  & 1.25   & 1.23   & 1.27  & 1.24    & 1.25    & 1.22   \\
    $n$=300            & 1.34 & 1.36  & 1.64   & 1.52   & 1.60  & 1.67    & 1.61    & 1.68   \\
    $n$=500            & 1.26 & 1.34  & 1.55   & 1.49   & 1.60  & 1.65    & 1.61    & 1.69   \\ \hline
    4-variate $t_{10}$ & SCM  & Tyler & $\tilde{\Sigma}$-H & $\tilde{\Sigma}$-M & $\tilde{\Sigma}$-P & ${\Sigma}_{*}$-H & ${\Sigma}_{*}$-M & ${\Sigma}_{*}$-P \\ \hline
    $n$=20             & 0.90 & 0.89  & 0.95   & 0.98   & 0.98  & 0.96    & 1.01    & 0.95   \\
    $n$=50             & 0.90 & 0.91  & 1.01   & 0.98   & 0.98  & 1.03    & 1.04    & 0.99   \\
    $n$=100            & 0.87 & 0.87  & 0.93   & 0.95   & 1.01  & 0.99    & 1.01    & 1.05   \\
    $n$=300            & 0.87 & 0.87  & 1.09   & 1.09   & 1.17  & 1.14    & 1.16    & 1.23   \\
    $n$=500            & 0.88 & 0.92  & 1.10   & 1.10   & 1.23  & 1.19    & 1.22    & 1.29   \\ \hline
    4-variate $t_{15}$  & SCM  & Tyler & $\tilde{\Sigma}$-H & $\tilde{\Sigma}$-M & $\tilde{\Sigma}$-P & ${\Sigma}_{*}$-H & ${\Sigma}_{*}$-M & ${\Sigma}_{*}$-P \\ \hline
    $n$=20             & 0.92 & 0.90  & 0.94   & 0.94   & 0.96  & 0.95    & 0.97    & 0.89   \\
    $n$=50             & 0.82 & 0.83  & 0.88   & 0.91   & 0.93  & 0.88    & 0.93    & 0.93   \\
    $n$=100            & 0.84 & 0.87  & 0.92   & 0.95   & 1.00  & 0.93    & 0.96    & 1.00   \\
    $n$=300            & 0.73 & 0.75  & 0.96   & 0.99   & 1.10  & 1.00    & 1.06    & 1.12   \\
    $n$=500            & 0.73 & 0.76  & 0.95   & 0.96   & 1.06  & 0.94    & 0.97    & 1.06   \\ \hline
    4-variate $t_{25}$  & SCM  & Tyler & $\tilde{\Sigma}$-H & $\tilde{\Sigma}$-M & $\tilde{\Sigma}$-P & ${\Sigma}_{*}$-H & ${\Sigma}_{*}$-M & ${\Sigma}_{*}$-P \\ \hline
    $n$=20             & 0.89 & 0.92  & 0.92   & 0.92   & 0.90  & 0.96    & 0.95    & 0.89   \\
    $n$=50             & 0.82 & 0.84  & 0.89   & 0.90   & 0.91  & 0.93    & 0.96    & 0.92   \\
    $n$=100            & 0.77 & 0.76  & 0.90   & 0.90   & 0.96  & 0.94    & 0.98    & 1.04   \\
    $n$=300            & 0.73 & 0.77  & 0.93   & 0.91   & 0.98  & 1.00    & 0.98    & 1.03   \\
    $n$=500            & 0.67 & 0.71  & 0.83   & 0.83   & 0.96  & 0.88    & 0.90    & 1.00   \\ \hline
    4-variate Normal   & SCM  & Tyler & $\tilde{\Sigma}$-H & $\tilde{\Sigma}$-M & $\tilde{\Sigma}$-P & ${\Sigma}_{*}$-H & ${\Sigma}_{*}$-M & ${\Sigma}_{*}$-P \\ \hline
    $n$=20             & 0.82 & 0.84  & 0.87   & 0.90   & 0.91  & 0.89    & 0.93    & 0.89   \\
    $n$=50             & 0.80 & 0.81  & 0.87   & 0.88   & 0.88  & 0.88    & 0.92    & 0.88   \\
    $n$=100            & 0.68 & 0.71  & 0.80   & 0.85   & 0.91  & 0.82    & 0.86    & 0.92   \\
    $n$=300            & 0.61 & 0.63  & 0.82   & 0.85   & 0.93  & 0.86    & 0.91    & 0.96   \\
    $n$=500            & 0.60 & 0.64  & 0.77   & 0.80   & 0.90  & 0.82    & 0.86    & 0.96   \\ \hline
    \end{tabular}
\end{scriptsize}
\caption{Finite sample efficiencies of estimators of the first eigenvector based on 
several scatter matrices in dimension $p=4$. The notation
 H, M or P after $\tilde{\Sigma}$ or ${\Sigma}_{*}$ indicates the depth function 
 used for the weights: H = halfspace depth, M = Mahalanobis depth, P = projection depth.}
\label{table:FSEtable4}
\end{table}

The results from this simulation exercise are presented in Table~\ref{table:FSEtable4}.
It can be seem that  $\tilde{\Sigma}$-based estimators (columns 3-5) 
outperform SCM and Tyler's $M$-estimator of scatter. Among the 3 depth functions 
considered, projection depth gives the best results. Its finite sample performances are 
better than Tyler's and Huber's M-estimators of scatter, as well as their symmetrized 
counterparts that are much more computationally intensive (see Table 4 in 
\cite{ref:JMVA071611_Sirkiaetal}). The affine equivariant ${\Sigma}_{*}$-based estimators (columns 6-8) are even more efficient.

\subsection{Influence function comparison}

\begin{figure}[t!]
	\centering
		\includegraphics[width=0.8\textwidth]{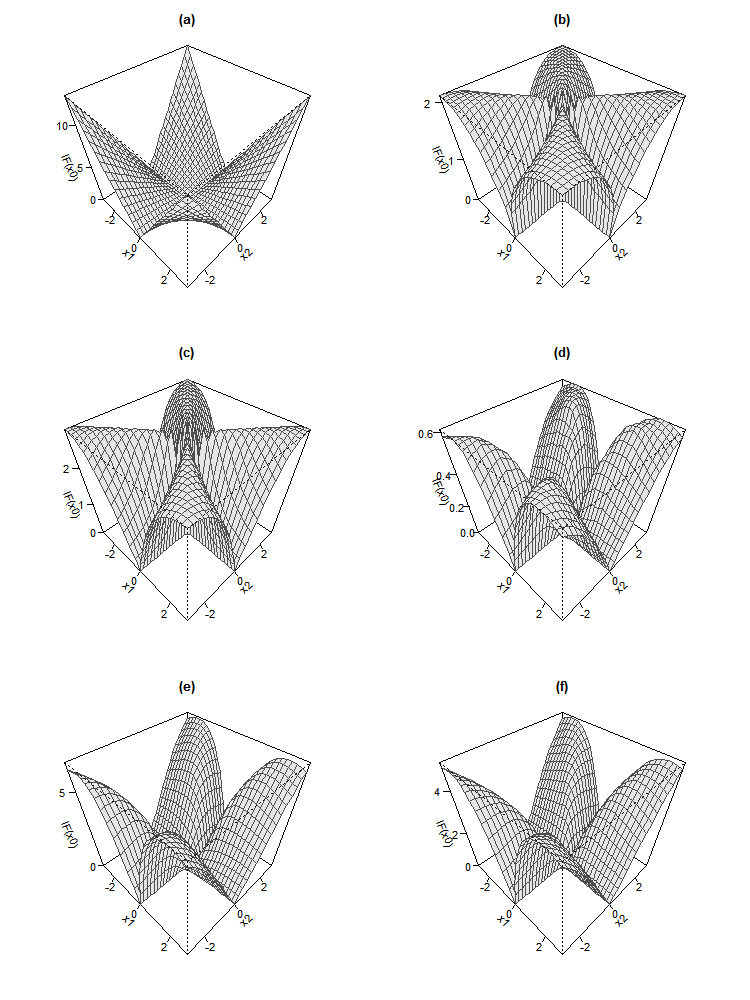}
	\caption{Plot of the norm of influence function for first eigenvector of (a) sample 
	covariance matrix, (b) SCM, (c) Tyler's scatter matrix and $\tilde{\Sigma}$ for 
	weights obtained from (d) Halfspace depth, (e) Mahalanobis depth, (f) Projection depth 
	for a bivariate normal distribution with $\bfmu = {\bf 0}, \Sigma = \diag(2,1)$}
	\label{fig:IFnorm}
\end{figure}

In Figure \ref{fig:IFnorm} we consider first eigenvectors of $\tilde{\Sigma}$, 
the Sign Covariance Matrix (SCM) 
and Tyler's shape matrix \citep{ref:AoS87234_Tyler}. We generate data from and set
$\BF \equiv \mathcal{N}_2({ 0}, \diag(2,1))$ and plot norms of the 
eigenvector influence functions for different values of $x_0$. 
Let us denote the $i$-th eigenvector of  the Sign Covariance Matrix and Tyler's shape 
matrix by  $\gamma_{S,i}$ and $\gamma_{T,i}$, respectively. 
Their influence functions are given  as follows:
\ban
 IF(x_0; \gamma_{S,i}, \BF) & = 
\sum_{k=1; k \neq i}^p \frac{\BS_{ik}(\Lambda^{1/2} z_0, { 0})}{\lambda_{S,i} - \lambda_{S,k}} \gamma_k; \\
& \text{ where }
\lambda_{S,i} = \BE_Z \left( \frac{\lambda_i z_i^2}{\sum_{j=1}^p \lambda_j z_j^2} \right),\\
IF(x_0; \gamma_{T,i}, \BF) & =  (p+2) \sum_{k=1; k \neq i}^p \frac{\sqrt{\lambda_i 
\lambda_k}}{\lambda_i - \lambda_k}\BS_{ik}(z_0; {0}) \gamma_k. 
\ean
Panels (b) and (c) in Figure~\ref{fig:IFnorm}, corresponding to  Sign Covariance Matrix 
and Tyler's shape matrix respectively, exhibit an `inlier effect', that is, 
points close to the center  having high influence,  which results in loss of efficiency. 
On the other hand, the influence function for eigenvector estimates of the 
sample covariance matrix  (panel (a)) is 
unbounded and makes the corresponding estimates non-robust. In comparison, the 
$\tilde{\Sigma}$ corresponding to weights derived from projection depth, 
half-space depth and Mahalanobis depth  have bounded influence functions {\it and} 
small values of the influence function at `deep' points.

\subsection{Efficiency of affine equivariant robust estimator}

We study the finite sample efficiency properties of $\Sigma_{*}$ 
using a simulation exercise. 
We consider 6 different elliptic distributions,  namely, the $p$-variate multivariate 
Normal distribution and the multivariate $t$ distributions corresponding to degrees 
of freedom 5, 6, 10, 15 and 25. We compute the ARE of the estimator for the 
first eigenvector using $\Sigma_{*}$, using weights based on the projection depth 
(PD) and the halfspace depth (HSD), thus this simulation is an illustration of how 
different choices of weights affect the results in the context of 
Theorem~\ref{Thm:Eigen2}.  We consider using the sample covariance 
matrix as the baseline method for this study. The ARE values are computed by
using Monte-Carlo simulation of $10^6$ samples and subsequent numerical integration.
We report the results of this exercise in Table~\ref{table:AREtable}.  
Based on these results, we notice that  $\Sigma_{*}$  is particularly efficient in lower 
dimensions for distributions with heavier tails ($t_5$ and $t_6$), while for distributions 
with lighter tails, the AREs increase with data dimension. At higher values of $p$,
note that $\Sigma_{*}$ is almost as efficient as the sample covarnace matrix even
 when the data comes from multivariate normal distribution.

\begin{table}[t]
\centering
\begin{footnotesize}
\begin{tabular}{c|cccc|cccc}
    \hline
    & \multicolumn{4}{c|}{PD} & \multicolumn{4}{c}{HSD} \\\cline{2-9}
    Distribution & $p=2$  & $p=5$  & $p=10$ & $p=20$ & $p=2$  & $p=5$  & $p=10$ & $p=20$ \\ \hline
    $t_5$           & 4.73 & 3.99 & 3.46 & 3.26 & 4.18 & 3.63 & 3.36 & 3.15 \\
    $t_6$           & 2.97 & 3.28 & 2.49 & 2.36 & 2.59 & 2.45 & 2.37 & 2.32 \\
    $t_{10}$          & 1.45 & 1.47 & 1.49 & 1.52 & 1.30 & 1.37 & 1.43 & 1.49 \\
    $t_{15}$          & 1.15 & 1.19 & 1.23 & 1.27 & 1.01 & 1.10 & 1.17 & 1.24 \\
    $t_{25}$          & 0.97 & 1.02 & 1.07 & 1.11 & 0.85 & 0.94 & 1.02 & 1.08 \\
    MVN          & 0.77 & 0.84 & 0.89 & 0.93 & 0.68 & 0.77 & 0.84 & 0.91 \\ 
    \hline
\end{tabular}
\end{footnotesize}
\caption{Table of AREs of the estimator for the first eigenvector estimation using 
$\Sigma_{*}$, relative to using the sample covariance matrix, for different choices of 
dimension $p$. The data-generating distributions are the multivariate Normal (MVN), 
and multivariate $t$-distributions with degrees of freedom 5, 6, 10, 15 and 25. Weights 
for $\Sigma_{*}$ are based on either the projection depth (PD) or the half-space 
depth (HSD).}
\label{table:AREtable}
\end{table}


\subsection{Robust sufficient dimension reduction and supervised learning}

One of the main usages of obtaining dispersion estimators and their eigenvalues and 
eigenvectors is in \textit{dimension reduction} techniques. Examples of such uses are in 
\textit{principal component regression, partial least squares} and \textit{envelope 
methods}. We illustrate below the latter technique, in the context of 
\textit{sufficient dimension reduction} (SDR). For details on envelope methods and other 
uses of robust estimators of dispersion and eigen-structures, see 
\citet{ref:Sinica10927_Cooketal, ref:PhilTransRoyalSoc094385_AdragniCook, 
ref:JASA15599_CookZhang} 
and references and citations of these 
articles. In the context of multivariate-response ($Y_{i} \in \BR^{q}$) linear regression, 
the envelope method proposes the model $Y_{i} = \alpha + \Gamma_{1} \eta x_{i} + e_{i}$, 
where $e_{i}$ are independent mean zero Gaussian  noise terms with covariance matrix 
$\Sigma$ whose spectral representation can be written as 
\ban 
\Sigma = \Gamma \Lambda \Gamma^{T}
& = 
\left( \begin{array}{ll}
\Gamma_{0} & \Gamma_{1} 
\end{array}
\right)
\left( \begin{array}{ll}
\Lambda_{0} & 0 \\
0 & \Lambda_{1} 
\end{array}
\right)
\left( \begin{array}{l}
\Gamma_{0} \\
\Gamma_{1} 
\end{array}
\right) \\
& = \Gamma_{0} \Lambda_{0} \Gamma_{0}^{T}
+ \Gamma_{1} \Lambda_{1} \Gamma_{1}^{T}.
\ean
Thus, the eigenvectors of $\Sigma$ are partitioned into two blocks:  
$\Gamma_{1} \in \BR^{q} \times \BR^{d}$ and 
$\Gamma_{0} \in \BR^{q} \times \BR^{q-d}$, and 
the regression coefficient of $Y_{i}$ on $x_{i}$  is given by 
$\Gamma_{1} \eta$ for some $\eta \in \BR^{d} \times \BR^{p}$. 
Dimension reduction is achieved  when $d \ll p$, typically without extraneous 
assumptions like sparsity. The envelope model for generalized linear models is 
discussed in \citet{ref:PhilTransRoyalSoc094385_AdragniCook}, 
and may be used for supervised 
learning. Nonlinear regression models may also be handled similarly.

Given a set of examples $\{ (Y_{i}, X_{i}), i = 1, \ldots, n \}$, an envelope-based 
prediction for the response $Y$ for any $X$ may be obtained from 
\ban
\hat{Y} (X) & = \bigl[ \sum_{i = 1}^{n} w_{i} \bigr]^{-1}
\sum_{i = 1}^{n} w_{i} Y_{i}, 
\text{ where } \\ 
 w_{i} &= \exp \left[ -\frac{1}{\hat\sigma^2}  | \hat\Gamma_{1}^T (X - X_i) |^2 \right].
\ean
The above assumes that the covariates come from the Gaussian distribution 
$N_{p} ({\bf 0}_{p}, \sigma^{2} \BI_{p})$, and appropriate 
changes may be made for other distributions. 

We design a robust version of the above, by using weighted spatial medians for location 
parameters corresponding to the distributions of $X$ and $X | Y$, and using the first $d$ 
eigenvectors of $\tilde{\Sigma}$ as $\hat\Gamma_{1}$. A robust location estimator for 
the distribution of $X | Y$ is required for the estimation of $\sigma^2$. 
Details are available in \cite{ref:PhilTransRoyalSoc094385_AdragniCook}.
\begin{figure}[t!]
\begin{center}
\begin{tabular}{ll}
\includegraphics[width=0.4\textwidth]{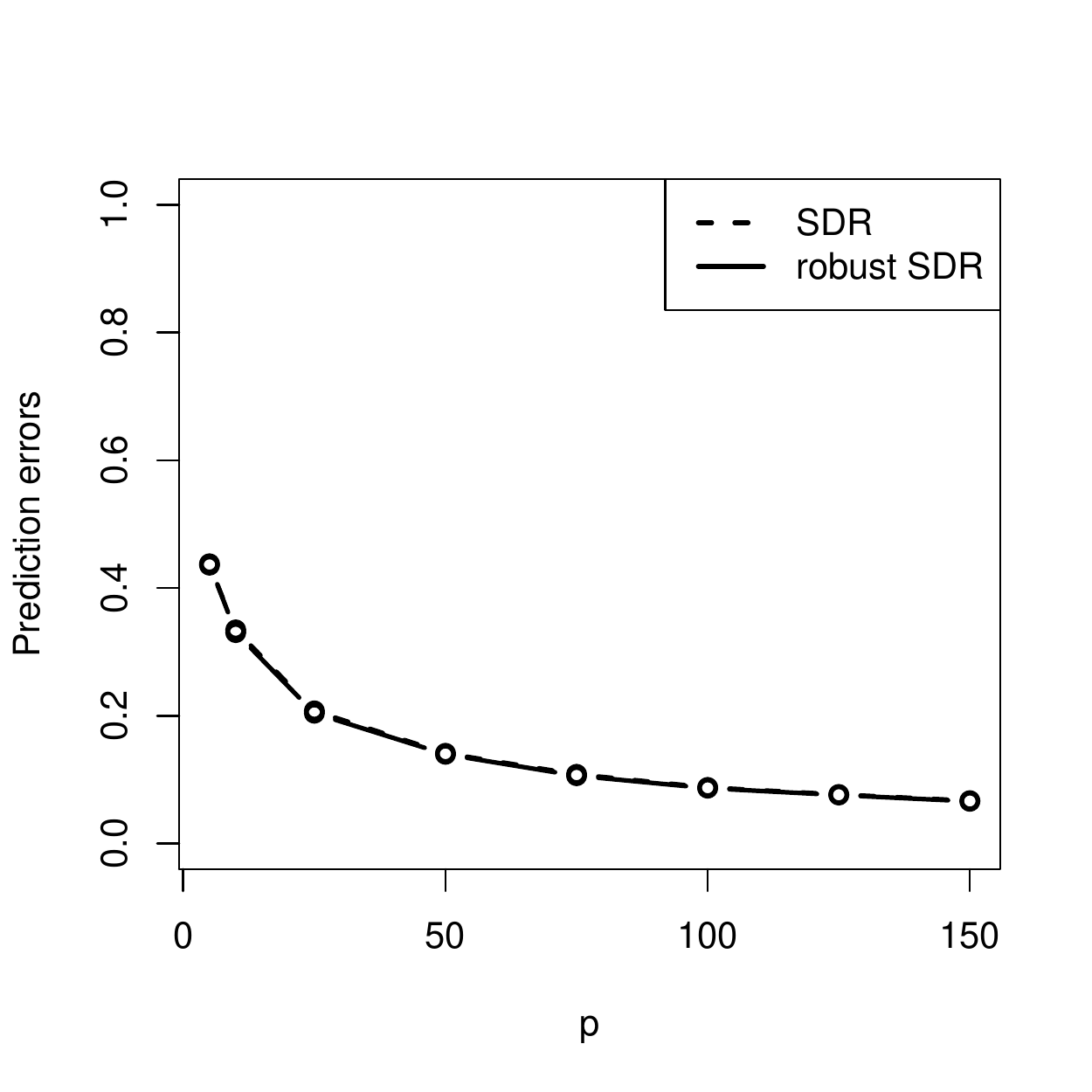} &
\includegraphics[width=0.4\textwidth]{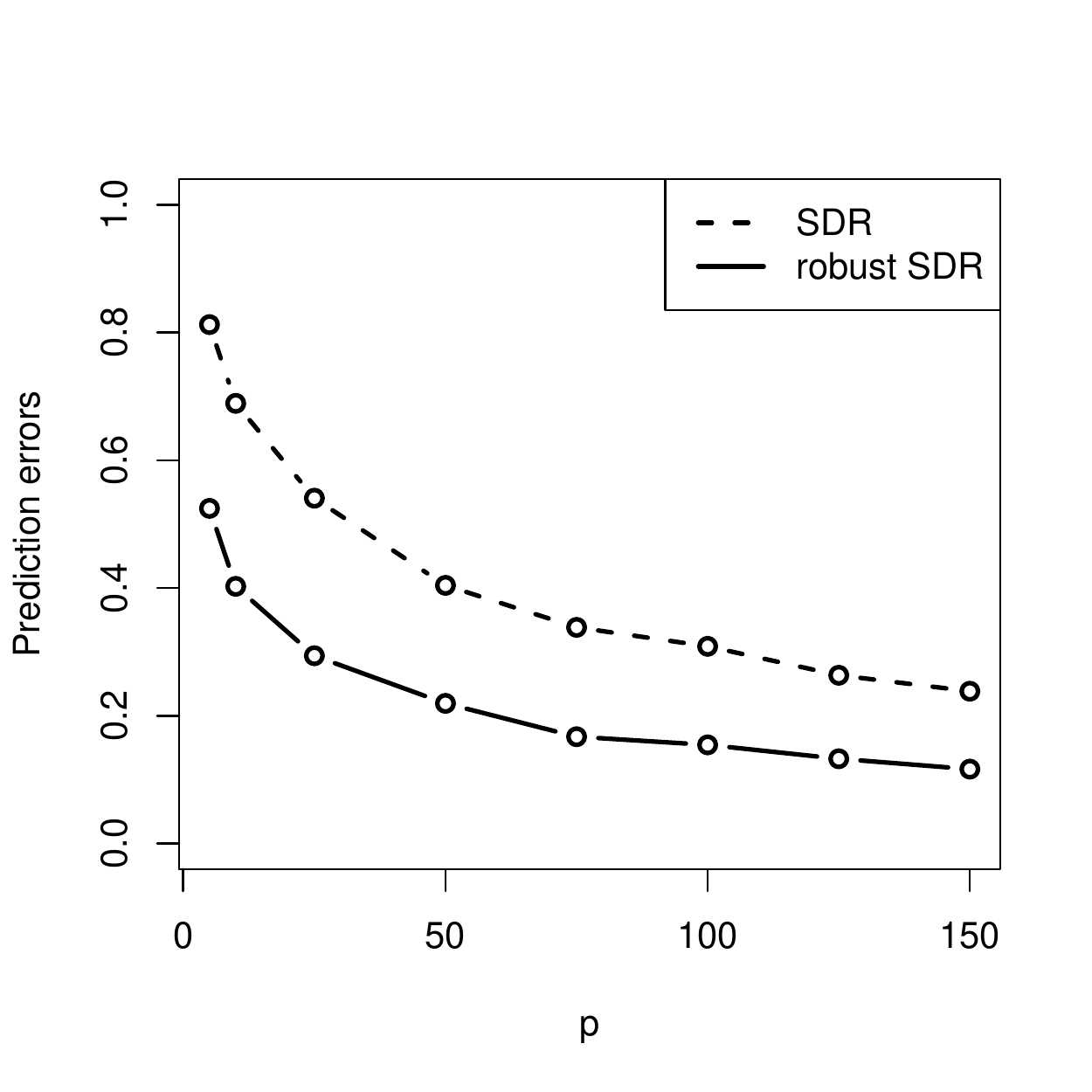}
\end{tabular}
\caption{Average prediction errors for two methods of SDR (a) in absence and (b) in presence of outliers}
\label{fig:SDRfig}
\end{center}
\end{figure}
In a non-linear regression model, we compare the performance of the robust version of 
SDR with the original method of 
\cite{ref:PhilTransRoyalSoc094385_AdragniCook} 
with or without the presence of bad leverage points in $\Sigma$. 
For any given choice of covariate dimension $p$, we take $n=200$ and $d=1$, 
and generate the responses 
$Y_1, \ldots, Y_n$ as independent standard normal, and 
$X | Y$ as Normal with mean $Y + Y^{2} + Y^3$ in each of the $p$ coordinates, 
and variance $25 \BI_p$.
We measure performances of the SDR models by their mean squared prediction error on 
another set of 200 observations  generated similarly, and taking the average of these 
errors on 100 such training-test pairs of datasets. The above steps 
are repeated for the choices of $p = 5, 10, 25, 50, 75, 100, 125, 150$.

Panel (a) of figure \ref{fig:SDRfig} compares prediction errors using both robust and 
maximum likelihood SDR estimates when the covariates contain no outliers: here the two 
methods are virtually indistinguishable. We then introduce outliers in each of the 100 
datasets by adding 100 to first $p/5$ coordinates of the first 10 observed covariate 
values, and repeat the analysis. Panel (b) of the figure shows that the robust SDR method 
remains more accurate in predicting out of sample observations for all values of $p$ than 
the standard SDR.

\section{Real Data Applications}
\label{Sec:DA}

We  now present an application of our proposed approach to some real data problems. 

Robust techniques are useful when in identifying outlying observations, and we illustrate 
below how to make use of our fixed-dimensional methods presented earlier for 
functional (and hence infinite-dimensional)  data. 

We follow the approach of \cite{ref:JASA151100_BoenteSalibianBarrera} 
for performing robust principal component 
analysis on functional data using the estimated eigenvectors from 
$\widehat{\tilde{\Sigma}}$. Suppose the data consistents of  $n$ curves, say 
$ \cF = \{ f_1, \ldots , f_n  \} \in L^2[0,1]$, each observed at a set of common design 
points $\{ t_1, ..., t_m \} $. We model each of these functions as a linear combination of 
$p$ mutually orthogonal B-spline basis functions $\cD = \{ \delta_1, ..., \delta_p \}$. We map data for each of the functions onto the coordinate system formed by the spline basis:
\begin{equation}
T( \cF, \cD)_{ij} = \sum_{l=2}^m f_i(t_l) \delta_j(t_l) (t_l - t_{l-1}); \quad 1 \leq i \leq n, 1 \leq j \leq p.
\end{equation}
We then model  the $i$-th row of the $n \times p$ matrix $T(\cF, \cD) \equiv T$, 
denoted by $\bfT_{i}$ 
as follows:
\ban 
\bfT_i = \mu + P s_i + e_{i},
\ean
where $\mu$ is a location parameter, $P$ is a $p \times q$ loading matrix, $s_{i}$ is 
a $q \times 1$ score vector, and $e_{i}$ is the error term. We obtain robust estimators of 
$\mu$, $P$ and consequently $s_{i}$ using $\widehat{\tilde{\Sigma}}$. Define 
$\widehat{\bfT}_{i} = \hat{\mu} + \hat{P} \hat{s}_{i}$. 
The \textit{orthogonal distance} (OD)  corresponding to this projection is defined as 
\ban 
OD_{i} = | \bfT_i - \widehat{\bfT}_{i} |.
\ean
Analogously, the \textit{score distance} (SD) is defined as 
\ban 
SD_i = \sqrt{ \sum_{j=1}^q \frac{\hat{s}^2_{ij}}{\hat \lambda_j}}; 
\ean
where $\hat \lambda_1,\ldots ,\hat \lambda_q$ are the top eigenvalues from 
$\widehat{\tilde{\Sigma}}$. 
For outlier detection, following \cite{ref:Technometrics0564_Hubertetal} 
we set the upper cutoff values for 
score distances at $(\chi^2_{2, .975})^{1/2}$ and orthogonal distances at 
\ban
[\text{median}(OD^{2/3}) + \text{MAD}(OD^{2/3})\Phi^{-1}(0.975)]^{3/2},
\ean
 where 
$\Phi(\cdot)$ is the standard normal cumulative distribution function.

\begin{figure}[t!]
\begin{center}
\includegraphics[width=0.9\textwidth]{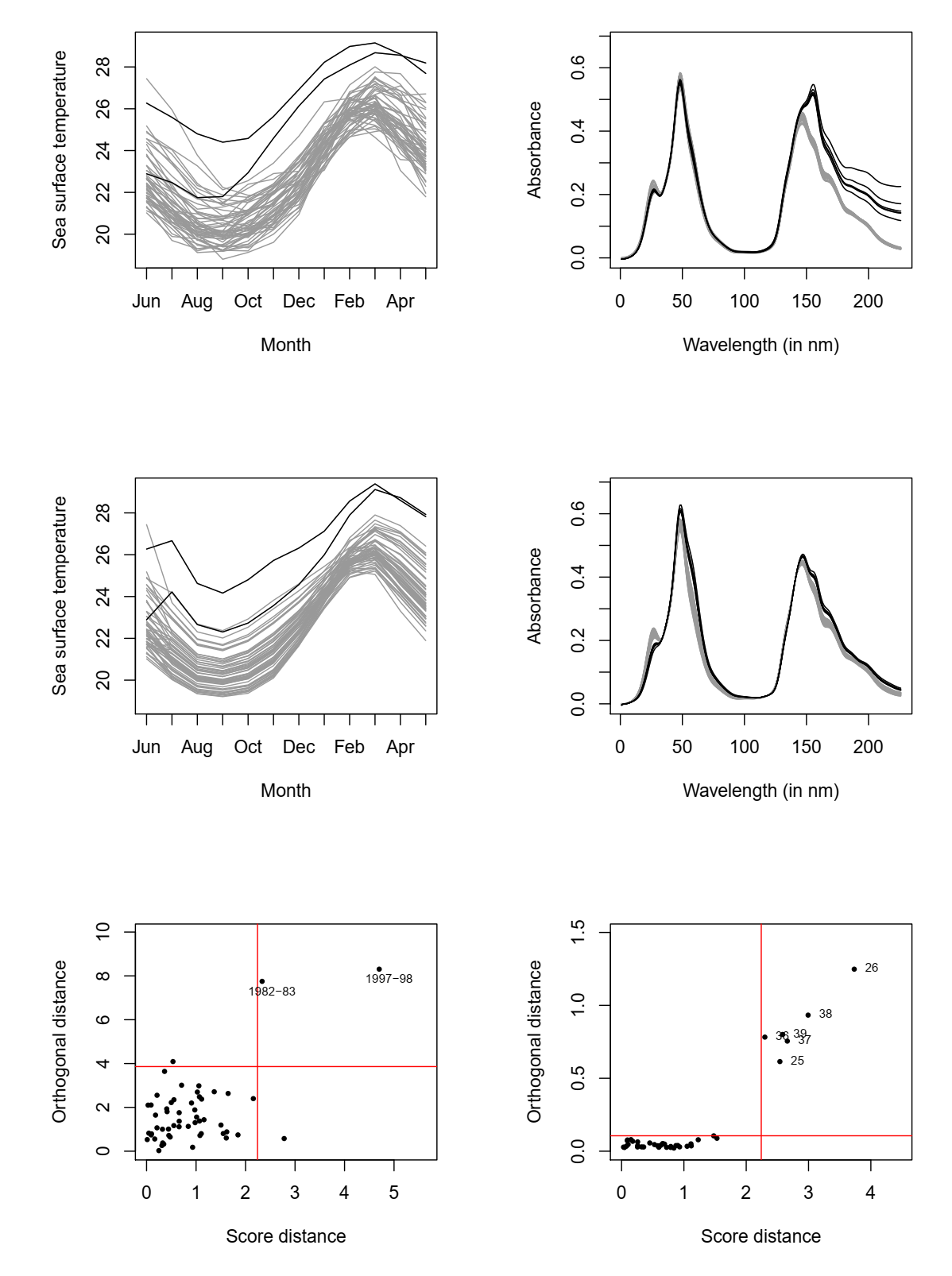}
\caption{Actual sample curves, their spline approximations and diagnostic plots respectively for El-Ni\~no (left side plots) and Octane (right side plots) datasets}
\label{fig:fPCAfig}
\end{center}
\end{figure}

We apply the above outlier detection method on two data sets. First, we consider the 
monthly average \textit{sea surface temperature anomaly data} 
from June 1970 to May 2004, available 
from \url{http://www.cpc.ncep.noaa.gov/data}, depicted in  top-left panel
of Figure~\ref{fig:fPCAfig})).
Second, we consider the \textit{octane data}, which consists of 226 variables and 39 
observations \citep{ref:EsbensenetalBook94}. 
Each sample is a gasoline compound with a certain octane 
number, and has its NIR absorbance spectra measured in 2 nm intervals between 1100 - 1550 
nm. There are 6 outliers here: compounds 25, 26 and 36-39, which contain alcohol. This 
data is presented in  top-right panel  of Figure~\ref{fig:fPCAfig})).

In the sea surface temperature data, using a cubic spline basis with knots at alternate 
months starting in June, we get a close approximation as depicted in middle-left
panel  of Figure~\ref{fig:fPCAfig}. Using our proposed methodology with $q =1$ 
results in two points having their SD and OD larger than cutoff, depicted in  
bottom-left panel  of Figure~\ref{fig:fPCAfig}. These points correspond to the time periods June 
1982 to May 1983 and June 1997 to May 1998 are marked by black curves in panels a and c, 
and pinpoint the two seasons with strongest El-Ni\~no events.

On the octane data, we use the same methodology, 
 and again the top robust PC turns out to be sufficient in identifying all 6 outliers. 
 Details are available in the right side panels  of Figure~\ref{fig:fPCAfig}.
\section{Conclusions}
\label{Sec:Conclusion}

We propose the use of a weighted multivariate sign transformation for robust 
estimation and inference, and as demonstrated by theoretical results and several 
simulation studies and data examples, in many situations using a data-depth driven weight 
function leads to considerable efficiency gain without compromising on robustness 
properties. Our methodology seems to suggest new ways of identifying 
El-Ni\~no or La-Ni\~na events from the sea-surface temperature anomaly data, 
which will be studied further later.

Several of our results stated above are for data from the Euclidean space $\BR^{p}$, where 
$p$ is fixed. The cases where $p$ increases with sample size and may be higher than sample 
size, and where data are from a separable Hilbert space, will be considered in a future 
work. There are only few conceptual challenges to such extensions, however, there are 
several  technical and algebraic challenges.

\section*{Acknowledgements}
The research of SC is partially  supported by the National Science Foundation (NSF) under grants \#~DMS-1622483, \#~DMS-1737918.

\bibliographystyle{apalike}
\bibliography{arxiv-main} 

\begin{thebibliography}{}

\bibitem[Adragni and Cook, 2009]{ref:PhilTransRoyalSoc094385_AdragniCook}
Adragni, K.~P. and Cook, R.~D. (2009).
\newblock Sufficient dimension reduction and prediction in regression.
\newblock {\em Philosophical Transactions of the Royal Society of London A:
  Mathematical, Physical and Engineering Sciences}, 367(1906):4385 -- 4405.

\bibitem[Anderson, 2009]{ref:AndersonBook09}
Anderson, T. (2009).
\newblock {\em An Introduction to Multivariate Statistical Analysis}.
\newblock Wiley, Hoboken, NJ, 3 edition.

\bibitem[Bali et~al., 2011]{ref:AoS112852_Balietal_PCA_Functional_Robust}
Bali, J.~L., Boente, G., Tyler, D.~E., and Wang, J.-L. (2011).
\newblock Robust functional principal components: {A} projection-pursuit
  approach.
\newblock {\em The Annals of Statistics}, 39(6):2852 -- 2882.

\bibitem[Boente and Salibian-Barrera,
  2015]{ref:JASA151100_BoenteSalibianBarrera}
Boente, G. and Salibian-Barrera, M. (2015).
\newblock $s$-estimators for functional principal component analysis.
\newblock {\em Journal of the American Statistical Association}, 110(511):1100
  -- 1111.

\bibitem[Cardot et~al., 2017]{ref:AoS17591_Cardotetal_Median_HilbertSpace}
Cardot, H., C{\'e}nac, P., and Godichon-Baggioni, A. (2017).
\newblock Online estimation of the geometric median in {H}ilbert spaces:
  {N}onasymptotic confidence balls.
\newblock {\em The Annals of Statistics}, 45(2):591 -- 614.

\bibitem[Chakraborty and Chaudhuri,
  2014]{ref:AoS141203_ChakrabortyChaudhuri_Banach_Quantile}
Chakraborty, A. and Chaudhuri, P. (2014).
\newblock The spatial distribution in infinite dimensional spaces and related
  quantiles and depths.
\newblock {\em The Annals of Statistics}, 42(3):1203 -- 1231.

\bibitem[Chatterjee and Bose, 2005]{ref:CBose_AoS05414}
Chatterjee, S. and Bose, A. (2005).
\newblock Generalized bootstrap for estimating equations.
\newblock {\em The Annals of Statistics}, 33(1):414 -- 436.

\bibitem[Chaudhuri, 1996]{ref:JASA96862_Chaudhuri}
Chaudhuri, P. (1996).
\newblock On a geometric notion of quantiles for multivariate data.
\newblock {\em Journal of the American Statistical Association}, 91(434):862 --
  872.

\bibitem[Cook et~al., 2010]{ref:Sinica10927_Cooketal}
Cook, R.~D., Li, B., and Chiaromonte, F. (2010).
\newblock Envelope models for parsimonious and efficient multivariate linear
  regression.
\newblock {\em Statistica Sinica}, pages 927 -- 960.

\bibitem[Cook and Zhang, 2015]{ref:JASA15599_CookZhang}
Cook, R.~D. and Zhang, X. (2015).
\newblock Foundations for envelope models and methods.
\newblock {\em Journal of the American Statistical Association}, 110(510):599
  -- 611.

\bibitem[Croux and Haesbroeck, 2000]{ref:Biometrika00603_CrouxHaesbroeck}
Croux, C. and Haesbroeck, G. (2000).
\newblock Principal component analysis based on robust estimators of the
  covariance or correlation matrix: Influence functions and efficiencies.
\newblock {\em Biometrika}, 87(3):603--618.

\bibitem[D\"{u}rre et~al., 2014]{durre14}
D\"{u}rre, A., Vogel, D., and Tyler, D. (2014).
\newblock The spatial sign covariance matrix with unknown location.
\newblock {\em J. Mult. Anal.}, 130:107--117.

\bibitem[Esbensen et~al., 1994]{ref:EsbensenetalBook94}
Esbensen, K.~H., Sch\"{o}nkopf, S., and Midtgaard, T. (1994).
\newblock {\em Multivariate Analysis in Practice}.
\newblock CAMO As, Trondheim, Germany.

\bibitem[Fang et~al., 1990]{ref:Fangetal90_Book}
Fang, K.-T., Kotz, S., and Ng, K.-W. (1990).
\newblock {\em Symmetric Multivariate and Related Distributions}.
\newblock CRC Press.

\bibitem[Haberman, 1989]{ref:AoS891631_Haberman}
Haberman, S.~J. (1989).
\newblock Concavity and estimation.
\newblock {\em The Annals of Statistics}, 17(4):1631 -- 1661.

\bibitem[Haldane, 1948]{ref:Biometrika48414_Haldane}
Haldane, J. B.~S. (1948).
\newblock Note on the median of a multivariate distribution.
\newblock {\em Biometrika}, 35(3 - 4):414 -- 417.

\bibitem[Hampel et~al., 1986]{ref:HampelBook86}
Hampel, F.~R., Ronchetti, E.~M., Rousseeuw, P.~J., and Stahel, W.~A. (1986).
\newblock {\em Robust Statistics: The Approach based on Influence Functions},
  volume 196.
\newblock John Wiley \& Sons.

\bibitem[Huber, 1981]{ref:HuberBook81}
Huber, P.~J. (1981).
\newblock {\em Robust Statistics}.
\newblock Wiley.

\bibitem[Hubert et~al., 2005]{ref:Technometrics0564_Hubertetal}
Hubert, M., Rousseeuw, P.~J., and Vanden~Branden, K. (2005).
\newblock {ROBPCA}: A new approach to robust principal component analysis.
\newblock {\em Technometrics}, 47(1):64 -- 79.

\bibitem[Koltchinskii, 1997]{ref:AoS97435_Koltchinskii}
Koltchinskii, V.~I. (1997).
\newblock \textit{M}-estimation, convexity and quantiles.
\newblock {\em The Annals of Statistics}, 25(2):435 -- 477.

\bibitem[Liu et~al., 1999]{LiuPareliusSingh99}
Liu, R., Parelius, J.~M., and Singh, K. (1999).
\newblock Multivariate analysis by data depth: descriptive statistics, graphics
  and inference.
\newblock {\em The Annals of Statistics}, 27(3):783--858.

\bibitem[Locantore et~al., 1999]{ref:Test991_Locantoreetal}
Locantore, N., Marron, J.~S., Simpson, D.~G., et~al. (1999).
\newblock Robust principal component analysis for functional data.
\newblock {\em Test}, 8(1):1 -- 73.

\bibitem[Magyar and Tyler, 2014]{ref:Biometrika14673_MagyarTyler}
Magyar, A.~F. and Tyler, D.~E. (2014).
\newblock The asymptotic inadmissibility of the spatial sign covariance matrix
  for elliptically symmetric distributions.
\newblock {\em Biometrika}, 101(3):673--688.

\bibitem[Majumdar and Chatterjee, 2018]{StatPaper18}
Majumdar, S. and Chatterjee, S. (2018).
\newblock Non-convex penalized multitask regression using data depth-based
  penalties.
\newblock {\em Stat}, 7:e174.

\bibitem[Miao and Ben-Israel,
  1992]{ref:LinearAlgebraApplications9281_MiaoBenIsrael}
Miao, J. and Ben-Israel, A. (1992).
\newblock On principal angles between subspaces in $\mathbb{R}^{n}$.
\newblock {\em Linear Algebra and its Applications}, 171:81 -- 98.

\bibitem[Minsker, 2015]{ref:Bernoulli152308_Minsker_Median_Banach}
Minsker, S. (2015).
\newblock Geometric median and robust estimation in {B}anach spaces.
\newblock {\em Bernoulli}, 21(4):2308 -- 2335.

\bibitem[M{\"o}tt{\"o}nen and Oja, 1995]{ref:JNonpara95201_MottonenOja95}
M{\"o}tt{\"o}nen, J. and Oja, H. (1995).
\newblock Multivariate spatial sign and rank methods.
\newblock {\em Journal of Nonparametric Statistics}, 5(2):201--213.

\bibitem[Niemiro, 1992]{ref:AoS921514_Niemiro}
Niemiro, W. (1992).
\newblock Asymptotics for {M}-estimators defined by convex minimization.
\newblock {\em The Annals of Statistics}, 20(3):1514--1533.

\bibitem[Oja, 2010]{ref:OjaBook10}
Oja, H. (2010).
\newblock {\em Multivariate Nonparametric Methods with R: An Approach Based on
  Spatial Signs and Ranks}.
\newblock Springer Science \& Business Media.

\bibitem[Serfling, 2006]{ref:DIMACS061_Serfling}
Serfling, R. (2006).
\newblock Depth functions in nonparametric multivariate inference.
\newblock {\em DIMACS Series in Discrete Mathematics and Theoretical Computer
  Science}, 72:1.

\bibitem[Sibson, 1979]{ref:JRSSB79217_Sibson}
Sibson, R. (1979).
\newblock Studies in the robustness of multidimensional scaling:
  {P}erturbational analysis of classical scaling.
\newblock {\em Journal of the Royal Statistical Society. Series B
  (Methodological)}, pages 217 -- 229.

\bibitem[Sirki{\"a} et~al., 2007]{ref:JMVA071611_Sirkiaetal}
Sirki{\"a}, S., Taskinen, S., and Oja, H. (2007).
\newblock Symmetrised $m$-estimators of multivariate scatter.
\newblock {\em Journal of Multivariate Analysis}, 98(8):1611 -- 1629.

\bibitem[Taskinen et~al., 2012]{ref:SPL12765_Taskinenetal}
Taskinen, S., Koch, I., and Oja, H. (2012).
\newblock Robustifying principal component analysis with spatial sign vectors.
\newblock {\em Statistics \& Probability Letters}, 82(4):765 -- 774.

\bibitem[Tyler, 1987]{ref:AoS87234_Tyler}
Tyler, D.~E. (1987).
\newblock A distribution-free $ m $-estimator of multivariate scatter.
\newblock {\em The Annals of Statistics}, 15(1):234 -- 251.

\bibitem[Wang et~al., 2015]{ref:JASA151658_WangPengLi}
Wang, L., Peng, B., and Li, R. (2015).
\newblock A high-dimensional nonparametric multivariate test for mean vector.
\newblock {\em Journal of the American Statistical Association}, 110(512):1658
  -- 1669.

\bibitem[Zuo and Serfling, 2000]{ref:AoS00461_ZuoSerfling}
Zuo, Y. and Serfling, R. (2000).
\newblock General notions of statistical depth function.
\newblock {\em The Annals of Statistics}, 28(2):461--482.

\end{thebibliography}

\appendix
\section*{Supplementary material}
\renewcommand{\thesection}{\Alph{section}}
\section{Form of $V_W$}\label{sec:appA}
First observe that for $F$ having covariance matrix $\Sigma = \Gamma\Lambda\Gamma^T$,
\begin{align}\label{eqn:VWeqn}
V_W  = (\Gamma \otimes \Gamma) V_{W, \Lambda} (\Gamma \otimes \Gamma)^T
\end{align}
where $V_{W, \Lambda}$ is the covariance matrix of $\BF_\Lambda$, the elliptic distribution with mean $\mu$ and covariance matrix $\Lambda$. Now,
\begin{eqnarray*}
V_{W, \Lambda} &=& \BE \left[ \ve \left\{ \frac{W^2 (Z, \BF_Z) \Lambda^{1/2} ZZ^T \Lambda^{1/2} }{ Z^T \Lambda Z} - \tilde \Lambda \right\} {\ve}^T \left\{ \frac{ W^2(Z, \BF_Z) \Lambda^{1/2} ZZ^T \Lambda^{1/2} }{ Z^T \Lambda Z} - \tilde \Lambda \right\} \right]\\
&=& \BE \left[ \ve \left\{ W^2 (Z, \BF_Z) \BS( \Lambda^{1/2} z; {\bf 0} ) \right\} {\ve}^T \left\{ W^2(Z, \BF_Z) \BS( \Lambda^{1/2}Z; \bf0) \right\} \right]\\
&& - \quad \ve( \tilde \Lambda) {\ve}^T( \tilde \Lambda )
\end{eqnarray*}

The matrix $\ve( \tilde\Lambda) \ve^T(\tilde\Lambda)$ consists of elements $\lambda_i\lambda_j$ at $(i,j)^\text{th}$ position of the $(i,j)^\text{th}$ block, and 0 otherwise. These positions correspond to variance and covariance components of on-diagonal elements. For the expectation matrix, all its elements are of the form $ \BE[ \sqrt{\lambda_a \lambda_b \lambda_c \lambda_d} Z_a Z_b Z_c Z_d . W^4 (Z, \BF_Z) / (Z^T \Lambda Z)^2]$, with $1 \leq a,b,c,d \leq p$. Since $W^4 (Z, \BF_Z) / (Z^T \Lambda Z)^2$ is even in $Z$, which has a spherically symmetric distribution, all such expectations will be 0 unless $a,b,c,d$ are all equal or pairwise equal. Following a similar derivation for spatial sign covariance matrices in \cite{ref:Biometrika14673_MagyarTyler}, we collect the non-zero elements and write the matrix of expectations:
$$ (\BI_{p^2} + \BK_{p,p}) \left\{ \sum_{a=1}^p \sum_{b=1}^p \tilde\gamma_{ab} (e_a e_a^T \otimes  e_b e_b^T) - \sum_{a=1}^p \tilde\gamma_{aa} (e_a e_a^T \otimes  e_a e_a^T) \right\} + \sum_{a=1}^p \sum_{b=1}^p \tilde\gamma_{ab} (e_a e_b^T \otimes  e_a e_b^T) $$
where $\BI_k = (e_1,...,e_k), \BK_{m,n} = \sum_{i=1}^m \sum_{j=1}^n \BJ_{ij} \otimes \BJ_{ij}^T$ with $\BJ_{ij} \in \BR^{m \times n}$ having 1 as $(i,j)$-th element and 0 elsewhere, and $\tilde\gamma_{mn} = \BE[ \lambda_m \lambda_n Z_m^2 Z_n ^2 . W^4 (Z, \BF_Z) / (Z^T \Lambda Z)^2]; 1 \leq m,n \leq p$.

\paragraph{}Putting everything together, denote by $\hat {\tilde \Lambda}$ the sample version of $\tilde \Lambda$, the weighted covariance matrix obtained from $\BF_\Lambda$, i.e. $\hat {\tilde \Lambda} = \sum_{i=1}^n W_n^2 (Z_i, \BF_{Z,n}) \BS( \Lambda^{1/2} Z_i; \hat \mu_n)/n $. Then the different types of elements in the matrix $\hat {\tilde \Lambda}$ are as given below ($1 \leq a,b,c,d \leq p$):

\begin{itemize}
\item Variance of on-diagonal elements
$$ A\BV( \sqrt n \hat {\tilde \Lambda} (a,a) ) = \BE \left[ \frac{ W^4(Z, \BF_Z) \lambda_a^2 Z_a^4}{(Z^T \Lambda Z)^2} \right] - \tilde\lambda_a^2 $$

\item Variance of off-diagonal elements ($a \neq b$)
$$ A\BV( \sqrt n \hat {\tilde \Lambda} (a,b) ) = \BE \left[ \frac{W^4 (Z, \BF_Z) \lambda_a \lambda_b Z_a^2 Z_b^2}{(Z^T \Lambda Z)^2} \right] $$

\item Covariance of two on-diagonal elements ($a \neq b$)
$$ A\BV(\sqrt n \hat {\tilde \Lambda} (a,a), \sqrt n \hat {\tilde \Lambda} (b,b) )
= \BE \left[ \frac{W^4 (Z, \BF_Z) \lambda_a \lambda_b Z_a^2 Z_b^2}{(Z^T \Lambda Z)^2} \right] - \tilde\lambda_{a} \tilde\lambda_{b} $$

\item Covariance of two off-diagonal elements ($a \neq b, c \neq d$)
$$ A\BV(\sqrt n \hat {\tilde \Lambda} (a,b), \sqrt n \hat {\tilde \Lambda} (c,d) ) = 0 $$

\item Covariance of one off-diagonal and one on-diagonal element ($a \neq b \neq c$)
$$ A\BV(\sqrt n \hat {\tilde \Lambda} (a,b), \sqrt n \hat {\tilde \Lambda} (c,c) ) = 0 $$
\end{itemize}

\noindent The above give all the elements of $V_{W,\Lambda}$. We plug these in \eqref{eqn:VWeqn} to recover $V_W$.

\section{Proofs}\label{section:appB}

\begin{proof}[Proof sketch of Corollary 2.4]
%

We start with slightly modified versions of Lemmas A.4 and A.5 in \cite{ref:JASA151658_WangPengLi}:

\begin{Lemma}\label{lemma:lemmaB1}
Given that condition 1 in Corollary 2.4 holds, we have $\lambda_{\max} (\Psi_{1W}) \leq 2 W_{\max} \frac{\lambda_{\max} (\Sigma) }{\text{Tr} (\Sigma)} (1+o(1))$.
\end{Lemma}

\begin{Lemma}\label{lemma:lemmaB2}
Define $\Psi_{3W} = \BE \left[ \frac{W^2(\epsilon, \BF_\epsilon)}{| \epsilon |^2} (\BI_p - S(\epsilon) S(\epsilon)^T )\right] $. Then $\lambda_{\max} (\Psi_{2W}) \leq W_{\max} \BE(| \epsilon |^{-1})$ and $\lambda_{\max} (\Psi_{3W}) \leq (W_{\max})^2 \BE( | \epsilon|^{-2})$. Further, if conditions 1 and 2 of Corollary 2.4 hold then $\lambda_{\min} (\Psi_{2W}) \geq \BE( W(\epsilon, \BF_\epsilon)/ | \epsilon |)(1+o(1))/\sqrt 3$.
\end{Lemma}
The lemmas can be proved using similar steps as the proofs of the original lemmas in \cite{ref:JASA151658_WangPengLi} and using the upper bound on the weight function. Corollary 2.4 is now proved by applying Corollary 2.2, plugging in upper bound of $\lambda_{\max} (\Psi_{1W})$ from Lemma~\ref{lemma:lemmaB1} and lower bound of $\lambda_{\max} (\Psi_{1W})$ from Lemma~\ref{lemma:lemmaB2} into the ARE expression.
\end{proof}

\begin{proof}[Proof of Theorem 3.4]
We suppose $G_n = (g_1, \ldots, g_p), L_n = \diag(l_1, \ldots, l_p)$. In spirit, this corollary is similar to Theorem 13.5.1 in \cite{ref:AndersonBook09}. We start with the following result, due to \citep{ref:SPL12765_Taskinenetal}, allows us to obtain asymptotic joint distributions of eigenvectors and eigenvalues of $\hat{\tilde \Sigma}$, provided we know the limiting distribution of $\hat{\tilde \Sigma}$ itself:

\begin{Theorem} \label{Theorem:decomp} 
Let $\BF_\Lambda$ be defined as before, and $\hat C$ be any positive definite symmetric $p \times p$ matrix such that at $F_\Lambda$ the limiting distribution of $\sqrt n \ve(\hat C - \Lambda)$ is a $p^2$-variate (singular) normal distribution with mean zero. Write the spectral decomposition of $\hat C$ as $\hat C = \hat P \hat\Lambda \hat P^T$. Then the limiting distributions of $\sqrt n \ve(\hat P - \BI_p)$ and $\sqrt n \ve( \hat\Lambda - \Lambda)$ are multivariate (singular) normal and
\begin{equation} \label{equation:decompEq}
\sqrt n \ve (\hat C - \Lambda)  = \left[ (\Lambda \otimes \BI_p) - (\BI_p \otimes \Lambda) \right] \sqrt n \ve (\hat P - \BI_p) + \sqrt n \ve (\hat \Lambda - \Lambda) + o_P(1)
\end{equation}
\end{Theorem}

The first matrix picks only off-diagonal elements of the left-hand side and the second one only diagonal elements. We shall now use this as well as the form of the asymptotic covariance matrix of the vectorized $\hat {\tilde \Sigma}$, i.e. $V_W$ to obtain limiting variance and covariances of eigenvalues and eigenvectors.

Due to the decomposition \eqref{equation:decompEq} we have, for $\BF_\Lambda$, the following relation between any off-diagonal element of $\hat{\tilde \Lambda}$ and the corresponding element in the estimate of eigenvectors, say $\hat {\tilde \Gamma}_\Lambda$:

$$ \sqrt n \hat {\tilde \gamma}_{\Lambda, ij} = \sqrt n \frac{\hat {\tilde \Lambda} (i,j) }{\tilde \lambda_i - \tilde \lambda_j}; \quad i \neq j$$

So that for eigenvector estimates of the original $\BF$ we have
\begin{equation}\label{equation:app1}
\sqrt n ( \hat{\tilde \gamma}_{i} - \gamma_i) =
\sqrt n \Gamma ( \hat {\tilde \gamma}_{\Lambda, i} - e_i ) =
\sqrt n \left[ \sum_{k=1; k \neq i}^p \hat {\tilde \gamma}_{\Lambda,i,k} \gamma_k + (\hat {\tilde \gamma}_{\Lambda,i,i} - 1) \gamma_i \right]
\end{equation}
Now $\sqrt n (\hat {\tilde \gamma}_{\Lambda,i,i} - 1) =  o_P(1)$ and $A\BV(\sqrt n \hat {\tilde \Lambda} (i,k), \sqrt n \hat {\tilde \Lambda} (i,l) ) = 0$ for $k \neq l$, so the above equation implies
$$ A\BV(g_i) = A \BV (\sqrt n ( \hat{\tilde \gamma}_{i} - \gamma_i)) = \sum_{k=1; k \neq i}^p \frac{A\BV(\sqrt n \hat {\tilde \Lambda} (i,k))}{(\tilde \lambda_i - \tilde \lambda_k)^2} \gamma_k \gamma_k^T $$

For the covariance terms, from \eqref{equation:app1} we get, for $i \neq j$,
\begin{eqnarray*}
A\BV(g_i, g_j) &=&
A\BV(\sqrt n ( \hat{\tilde \gamma}_{i} - \gamma_i),
\sqrt n ( \hat{\tilde \gamma}_{j} - \gamma_j))\\
&=& A\BV \left(
\sum_{k=1; k \neq i}^p \sqrt n \hat {\tilde \gamma}_{\Lambda,ik} \gamma_k,
\sum_{k=1; k \neq j}^p \sqrt n \hat {\tilde \gamma}_{\Lambda,jk} \gamma_k \right)\\
&=& A\BV \left(
\sqrt n \hat {\tilde \gamma}_{\Lambda,ik} \gamma_j,
\sqrt n \hat {\tilde \gamma}_{\Lambda,ik} \gamma_i \right)\\
&=& - \frac{A\BV(\sqrt n \hat {\tilde \Lambda} (i,j))}
{(\tilde \lambda_i - \tilde \lambda_j)^2} \gamma_j \gamma_i^T
\end{eqnarray*}
The exact forms given in the statement of the corollary now follows from the  Form of $V_W$ in Section~\ref{sec:appA}.

\paragraph{}For the on-diagonal elements of $\hat {\tilde \Lambda}$, using Theorem \ref{Theorem:decomp} we have for the $i^\text{th}$ eigenvalue of $\hat {\tilde \Lambda}$, say $\lambda_{\Lambda,i}$,
$$ \sqrt n \hat {\tilde \lambda}_{\Lambda,i} = \sqrt n \hat {\tilde \Lambda} (i,i), $$
for $i = 1,...,p$. Hence

\begin{eqnarray*}
A\BV(l_i) &=& A\BV( \sqrt n (\hat {\tilde \lambda}_{\Lambda,i} - \tilde \lambda_i) )\\
&=& A\BV( \sqrt n (\hat {\tilde \lambda}_{\Lambda,i} - \tilde \lambda_{\Lambda,i}) )\\
&=& A\BV(\sqrt n \hat {\tilde \Lambda} (i,i))
\end{eqnarray*}

A similar derivation gives the expression for $A\BV(l_i,l_j); i \neq j$. Finally, since the asymptotic covariance between an on-diagonal and an off-diagonal element of $\hat {\tilde \Lambda}$, it follows that the elements of $G_n$ and diagonal elements of $L_n$ are independent.
\end{proof}
\end{document}